\documentclass{amsart}

\usepackage{graphicx}
\usepackage{amssymb}
\usepackage{amsmath}
\usepackage{amsfonts}
\usepackage[obeyFinal,textsize=footnotesize]{todonotes}

\usepackage{array}
\usepackage{multirow}
\usepackage{makecell}
\usepackage{arydshln}
\usepackage{multirow,bigdelim}
\usepackage{mathtools}
\usepackage{verbatim}
\usepackage{tabu}
\usepackage[font=small]{caption}
\usepackage{url}
\usepackage{hhline}
\usepackage{fancyvrb}
\usepackage{comment}
\usepackage{hhline}

\usepackage{blkarray}

\newtheorem{theorem}{Theorem}[section]

\newtheorem{lemma}[theorem]{Lemma}

\theoremstyle{definition}

\newtheorem{example}[theorem]{Example}

\theoremstyle{remark}
\newtheorem{remark}[theorem]{Remark}

\numberwithin{equation}{section}
\newcommand{\RN}[1]{%
  \textup{\uppercase\expandafter{\romannumeral#1}}%
}
\newcommand{\rn}[1]{%
  {\lowercase\expandafter{(\romannumeral#1)}}%
}
\newcommand{\gl}[1]{\text{GL}(#1)}
\newcommand{\ortho}[1]{\text{O}(#1)}
\newcommand{\un}[1]{\text{U}(#1)}
\newcommand{\symp}[1]{\text{Sp}(#1)}
\DeclareMathOperator{\tr}{tr}
\DeclareMathOperator{\ad}{ad} 
\DeclareMathOperator{\lie}{Lie}

\newcommand{\kak}{\text{K}_1\text{AK}_2}
\newcommand{\twoone}[2]{\begin{bmatrix}#1 \\ #2\end{bmatrix}}
\newcommand{\twotwo}[4]{\begin{bmatrix}#1 & #2\\#3 & #4\end{bmatrix}}
\newcommand{\twotworr}[4]{
\left[ \begin{array}{rr}
#1 & #2\\#3 & #4
\end{array} \right] }

\newcommand{\mi}{\!\scalebox{0.75}[1.0]{$-$}}
\newcommand{\mis}{\scalebox{0.3}[0.7]{\( - \)}}

\makeatletter
\renewcommand*\env@matrix[1][*\c@MaxMatrixCols c]{%
  \hskip -\arraycolsep
  \let\@ifnextchar\new@ifnextchar
  \array{#1}}
\makeatother

\begin{document}

\title[On the Cartan Decomposition for Classical Ensembles]{On the Cartan
Decomposition for Classical Random Matrix Ensembles}

\author{Alan Edelman}
\address{Department of Mathematics and Computer Science \& AI Laboratory, Massachusetts Institute of Technology, Cambridge, Massachusetts 02139}
\email{edelman@mit.edu}
\author{Sungwoo Jeong}
\address{Department of Mathematics, Massachusetts Institute of Technology, Cambridge, Massachusetts 02139}
\email{sw2030@mit.edu}

%    General info
%\subjclass[2010]{Primary 15A23, 15B52; Secondary 22E60, 53C35}

%\date{November 15, 2020.}

\begin{abstract}
We complete Dyson's dream by cementing the links between symmetric spaces and classical random matrix ensembles. Previous work has focused on a one-to-one correspondence between symmetric spaces and many but not all of the classical random matrix ensembles. This work shows that we can completely capture all of the classical random matrix ensembles from Cartan's symmetric spaces through the use of alternative coordinate systems. In the end, we have to let go of the notion of a one-to-one correspondence. 

\par We emphasize that the KAK decomposition traditionally favored by mathematicians is merely one coordinate system on the symmetric space, albeit a beautiful one. However, other matrix factorizations, especially the generalized singular value decomposition from numerical linear algebra, reveal themselves to be perfectly valid coordinate systems that one symmetric space can lead to many classical random matrix theories.

\par We establish the connection between this numerical linear algebra viewpoint and the theory of generalized Cartan decompositions. This in turn allows us to produce yet more random matrix theories from a single symmetric space. Yet again these random matrix theories arise from matrix factorizations, though ones that we are not aware have appeared in the literature.

% We present a completed connection between Cartan's Riemannian symmetric spaces and the classical random matrix ensembles: Hermite, Laguerre, Jacobi and circular. Previous authors associate a symmetric space to a random matrix ensemble by the decomposition of symmetric spaces studied by Cartan. However this is incomplete particularly in the compact case, in that the parameters of the Jacobi ensemble are not fully covered. 

% \par We propose the generalized Cartan decomposition of symmetric spaces as a tool for associating symmetric spaces to random matrix theories. Indeed, the generalized Cartan approach fully covers the classical Jacobi ensemble parameters by letting go of Cartan's special coordinates and imposing multiple coordinate systems on a single symmetric space. We also present new families of the Jacobi ensemble parameters by this approach. 

\keywords{Classical random matrix ensemble, Matrix factorization, Symmetric space, Generalized Cartan decomposition}
\end{abstract}

\maketitle
\VerbatimFootnotes

\section{Introduction}

\par Random matrix theory (RMT) is a big subject touching so many fields of mathematics, science, and engineering. For such a big subject it is helpful to have a means of cataloging the objects to be studied and a theory that covers the objects in the catalog. In 1962, Freeman Dyson \cite{dyson1962a,dyson1962b,dyson1962c,dyson1962threefold} was the first to propose a systematic approach to RMT. In the beginning of \cite{dyson1962threefold} he states his noble intent:
\begin{quote}
    \textit{To bring together and unify three trends of thought which have grown up independently during the last thirty years.}
\end{quote}
which he enumerates as \rn{1} group representations including time-inversion \rn{2} Weyl's theory of matrix algebras and \rn{3} RMT.

\par Around a decade later, Dyson hit upon the idea that symmetric spaces should play a key role \cite[Section V]{dyson1970correlations}. Dyson's suggestion was taken up in famous papers by Zirnbauer \cite{altland1997nonstandard,zirnbauer1996riemannian} and others \cite{caselle1996new,ivanov2002random}. These papers mainly focus on the noncompact cases. On the mathematical side, inspired by Katz and Sarnak \cite{katz1999random,katz1999zeroes}, Due\~{n}ez detailed connections to RMT for the compact symmetric spaces \cite{duenezthesis,duenez2004random}.

\par Nonetheless we felt there was a gap. When one juxtaposes \rn{1} the well established theory of classical random matrix ensembles with \rn{2} the RMTs associated with symmetric spaces, ensembles are missing. In particular, only very special Jacobi ensembles (the left side of Figure \ref{fig:newjacobi}) seem to be making the symmetric space list. More precisely, if one starts with a symmetric space, one has to make what we call a coordinate system choice, what others might call a matrix factorization choice. This choice has been the map $\Phi : K\times A\to G/K; (k, a)\mapsto kaK$ of Cartan, which we could call the KAK decomposition.\footnote{Although it is often called Cartan's KAK decomposition, Cartan was not aware of $G = KAK$.} See Figure \ref{fig:overview}.% The result is that not all classical Jacobi ensembles are obtained. 

\begin{figure}[h]
    \centering
    \fbox{\includegraphics[width=4.9in]{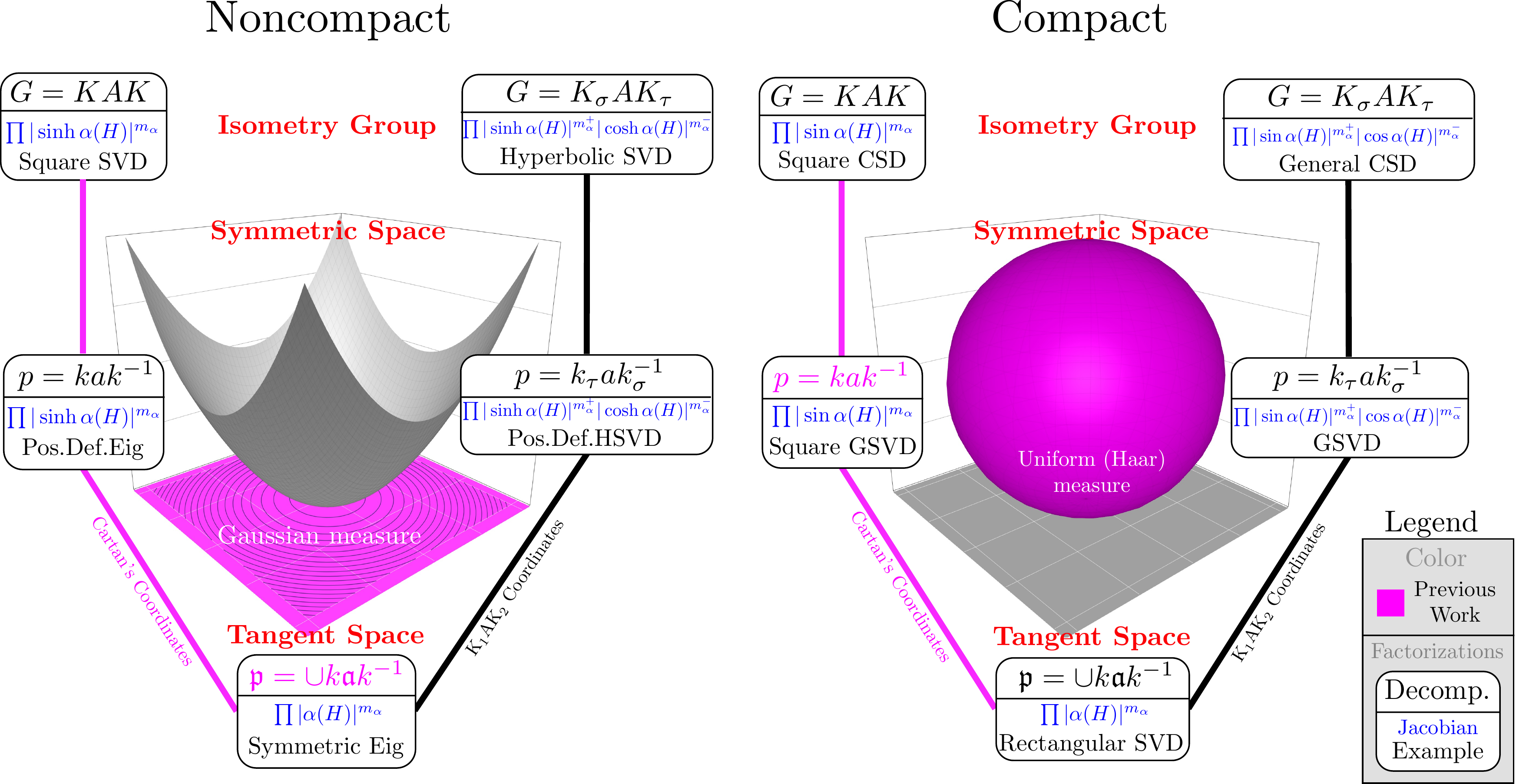}}
    \caption{Families of matrix factorizations associated with a symmetric space, its tangent space, and its isometry group: Shown above are the five factorizations skeleton associated with noncompact (left) and compact (right) symmetric spaces. Each serve as coordinate systems on the respective manifolds. {\color{violet} Previous approaches (manifold, coordinate system, and measure) are shown in magenta.} Examples of the linked factorizations/coordinate systems are shown.}
    \label{fig:overview}
\end{figure}

\par We show that coordinate systems from the generalized Cartan ($\kak$) decomposition associate a single symmetric space to multiple RMTs. Letting go of the historical bias of the KAK decomposition, the full set of Jacobi ensembles (the right side of Figure \ref{fig:newjacobi}) emerges, thereby leading to the complete list of classical random matrix ensembles. Of course, there is much mathematical precedent in differential geometry to letting go of any one special coordinate system. 

%In this work, we show that previous researchers have implicitly favored a particular coordinate system for symmetric spaces, the one familiar from Cartan's theory and the KAK decomposition\footnote{Although it is often called Cartan's KAK decomposition, Cartan did not discover $G = KAK$.} (this coordinate system is sometimes called the polar coordinate decomposition, e.g., see Helgason \cite[p.402]{Helgason1978}). Letting go of the historical bias of the KAK decomposition, the full set of Jacobi ensembles emerge. Of course, there is much mathematical precedent in differential geometry to letting go of any one special coordinate system. 

\begin{figure}[h]
    \centering
    { \large Possible parameters $(\alpha_1, \alpha_2)$ of the $\beta=2$ Jacobi ensemble\\ \large\vspace{-0.6cm} $$J_{\alpha_1, \alpha_2}(x) \sim \prod_{j<k}|x_j - x_k|^2 \prod_{j=1}^q x_j^{\alpha_1}(1-x_j)^{\alpha_2}$$\vspace{-0.2cm}\\}
    \includegraphics[width=4in]{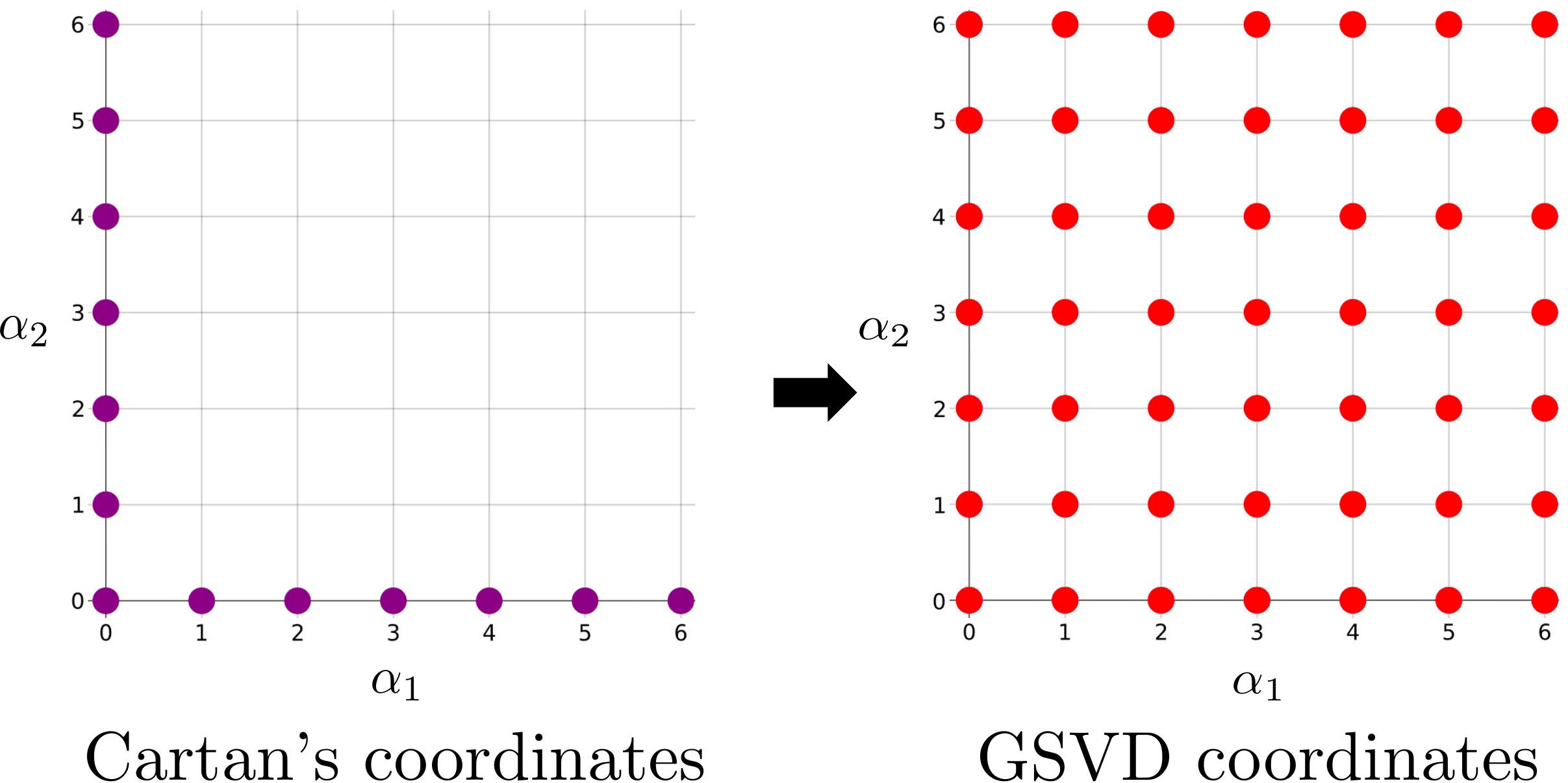}
    \caption{The parameter space $(\alpha_1, \alpha_2)\in(-1,\infty)^2$ of the $\beta=2$ Jacobi ensemble obtained from Cartan's coordinates (KAK) (left) and the generalized 
    singular value decomposition coordinates ($\kak$) (right).}
    \label{fig:newjacobi}
\end{figure}

\subsection{Classical Random Matrix Ensembles}\label{sec:classicalRMT}

The objects that we are interested in are the classical random matrix ensembles. Well established conventions\footnote{The term ``classical random matrix ensembles" may be found in well-known references:
\begin{itemize}
    \item Chapter 1 of Forrester's paper \cite{forrester1998random} has the title ``Classical Random Matrix Ensembles," and the even sections (1.2, 1.4, 1.6, 1.8) are explicitly Hermite, circular, Laguerre, Jacobi in that order. (Odd sections have discussions related to these ensembles.) Forrester's comprehensive book \cite{forrester2010log} deals exclusively Hermite, Laguerre, Jacobi and circular ensembles in Chapter 1-3 where the preface states: ``eigenvalue p.d.f. of the various classical $\beta$-ensembles given in Chapter 1-3." Then, later in Chapter 5.4, he further justifies the terminology by pointing out the four weights from classical orthogonal polynomial theory. 
    \item Anderson, Guionnet, Zeitouni \cite{anderson2010introduction}: Chapter 4.1 is entitled ``Joint distribution of eigenvalues in the classical matrix ensembles" and specifically covers exactly the Hermite, Laguerre, Jacobi, and circular ensembles.
    \item The first author's 2005 \textit{Acta Numerica} article \cite{edelman2005random}, Section 4.
\end{itemize}}
in random matrix theory agree that the ensembles in this class consist of the Hermite, Laguerre, Jacobi and circular ensembles built from matrices of integer sizes and involve entries that are real, complex, or quaternion. (What Dyson denoted $\beta=1,2,4$, and other authors in mathematics denote $\alpha=2/\beta=2,1,1/2$.)

% We can expand the catalog by taking on any weight function $w(x)\ge 0$ and $\beta$ which can take on the continuum of values (positive, say \cite{edelman2008beta}) and consider distributions of the form 
% \begin{equation*}
%     C \prod_{j<k} |x_j-x_k|^\beta \prod_j w(x_j) 
% \end{equation*}
% over the product of the intervals where $w$ is defined. (Hence the names Hermite, Laguerre, Jacobi, and circular.

% as we will develop here that the \rn{3} missing Jacobi ensembles can be associated with multiple\sung{"alternative" instead of "multiple"?} coordinate systems on the same Grassmannian manifold which is a symmetric space. From this point of view, we conclude not that the Jacobi's are missing, but rather there is not one but multiple random matrix ensembles associated with these symmetric spaces, depending on your choice of coordinate system. This observation may be viewed in the context of the $\kak$ decomposition, thereby completing the classification of classical matrix ensembles following Dyson's dream.

% We can expand the catalog by taking on any weight function $w(x)\ge 0$ and $\beta$ which can take on the continuum of values (positive, say \cite{edelman2008beta}) and consider distributions of the form 
% \begin{equation*}
%     C \prod_{j<k} |x_j-x_k|^\beta \prod_j w(x_j) 
% \end{equation*}
% over the product of the intervals where $w$ is defined. (Hence the names Hermite, Laguerre, Jacobi, and circular.) 

\par If one starts with the list of ten infinite families of Cartan's symmetric spaces\footnote{We will not discuss finite families of the exceptional types.} and asks to characterize which classical random matrix ensembles are covered, answers could be found in \cite[Table 1]{caselle1996new}, \cite[Table 1]{ivanov2002random} (noncompact cases) and \cite[Table 1]{duenez2004random} (compact cases). However, turning the question around,  if one starts with the classical random matrix ensembles, and asks whether symmetric spaces are adequate to explain all of them, we find the answer is a big ``almost", as the Jacobi ensembles are not adequately covered. To be precise, the Jacobi densities associated to compact symmetric spaces BD$\RN{1}$, A$\RN{3}$, C$\RN{2}$ from the previous attempts by the KAK decomposition are the following joint probability densities with $\beta=1,2,4$ (up to constant) and integers $p\geq q$:
\begin{equation}\label{eq:jacobiincompleteparameter}
    \text{KAK decomposition}:\hspace{.2cm}\prod_{j<k}|x_j - x_k|^\beta\displaystyle\prod_{j=1}^q x_j^{\frac{\beta}{2} - 1}(1-x_j)^{\frac{\beta(p-q+1)}{2} - 1},
\end{equation}
which we observe the powers of $x_j$'s restricted to $\frac{\beta}{2}-1$. The possible parameters of \eqref{eq:jacobiincompleteparameter} are described in the left side of Figure \ref{fig:newjacobi}. Additional four compact symmetric spaces D$\RN{3}$, BD, C, C$\RN{1}$ add four more Jacobi ensembles \cite{duenez2004random} but they are not sufficient to cover the two dimensional parameter set of the Jacobi ensembles. 

\subsection{Coordinate systems on the Grassmannian manifold}\label{sec:introGSVD}

It is always interesting when a branch of applied mathematics reverses direction and provides guidance to pure mathematics. In this work, we focus on the role of the generalized singular value decomposition (GSVD) from numerical linear algebra \cite{paige1981gsvd,van1976gsvd}. 

\par From an applied viewpoint, the Jacobi ensembles are elegantly generated in software with commands such as \verb|svdvals(randn(p,s),randn(q,s))| in languages such as Julia, which is computed by taking the GSVD of two i.i.d.\ normal matrices with the same number of columns \cite{edelman2008beta,edelman2018random}. From a pure viewpoint this is a pushforward of the uniform measure on the Grassmannian manifold
onto a maximal abelian subgroup $A$ (with a fixed Weyl chamber) along the generalized Cartan ($\kak$) decomposition \cite{flensted1978spherical,hoogenboom1983generalized}. 

\begin{figure}[h]
    \centering
    \fbox{\includegraphics[width=3.4in]{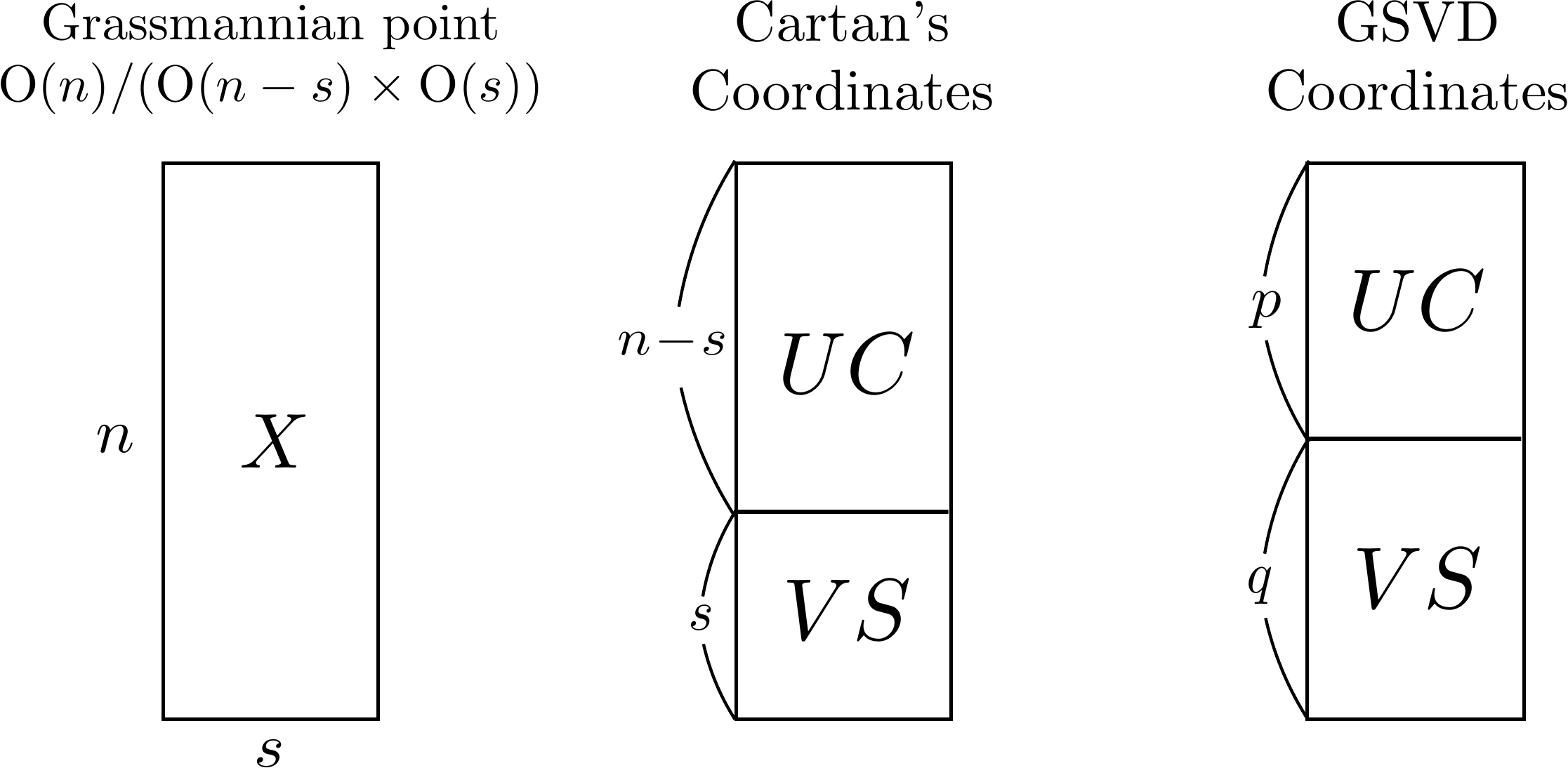}}
    \vspace{-0.3cm}\\\caption{Cartan's coordinate system (KAK) and GSVD coordinate systems ($\kak$) on the Grassmannian manifold $\ortho{n}/(\ortho{n-s}\times\ortho{s})$ }
    \label{fig:gsvdcoordinates}
\end{figure}

\par For example, take a Grassmannian point with any $\beta=1,2,4$ from $\ortho{n}/(\ortho{n-s}\times\ortho{s})$ (resp. with complex or quaternionic unitary groups) and represent it by the $n\times s$ orthogonal (resp. complex or quaternionic unitary) matrix\footnote{View the Grassmannian manifold as the quotient $V_s(\mathbb{R}^n)/\ortho{s}$ where $V_s(\mathbb{R}^n)$ is the Stiefel manifold. We are allowed to multiply any $O\in\ortho{s}$ on the right side of $X$.} $X$. For any $p, q \geq s$ satisfying $p+q = n$, we have the following coordinate system of $X$ arising from the GSVD\footnote{Alternatively, one can imagine the partial format of the CS decomposition. This is also equivalent to the bi-Stiefel decomposition with another quotient by the orthogonal group on the right.} \cite{edelman2020gsvd} of the first $p$ rows and the last $q$ rows of $X$:
\begin{equation}\label{eq:introgsvdcoordinates}
    X = \begin{bmatrix} U & \\ & V \end{bmatrix}\begin{bmatrix}C \\ S\end{bmatrix} = \begin{bmatrix}UC\\VS\end{bmatrix}.
\end{equation}
where $U, V$ are $p\times s$, $q\times s$ orthogonal (resp. complex or quaternionic unitary) matrices and $C, S$ are $s\times s$ diagonal matrices with cosine and sine values. Deduced joint probability densities \cite{edelman2018random} ($p, q\geq s$) are the following: (up to constant)
\begin{equation*}
    \text{$\kak$ decomposition (GSVD)}:\hspace{0.2cm}\prod_{j<k}|x_j - x_k|^\beta\prod_{j=1}^s x_j^{\frac{\beta(q-s+1)}{2}-1} (1-x_j)^{\frac{\beta(p-s+1)}{2}-1},
\end{equation*}
where the case $q=s$ represents the usual KAK decomposition case \eqref{eq:jacobiincompleteparameter}. 

As can be seen, the classical Jacobi parameters are quantized as they are integer multiples of $\beta/2$. Random matrix models that remove this quantization, thereby going beyond the classical, appear in \cite{dumitriu2002matrix,edelman2008beta,killip2004matrix}. In Section \ref{sec:mixed}, we also illustrate that some Jacobi ensembles can arise from symmetric spaces that are outside the traditional quantization (Figure \ref{fig:newjacobi2}).

\subsection{Contributions of this paper}

\begin{figure}[h]
    \centering
    {\large Example: Two symmetric spaces and Six corresponding RMTs\\[3pt]}
    \fbox{\includegraphics[width=4in]{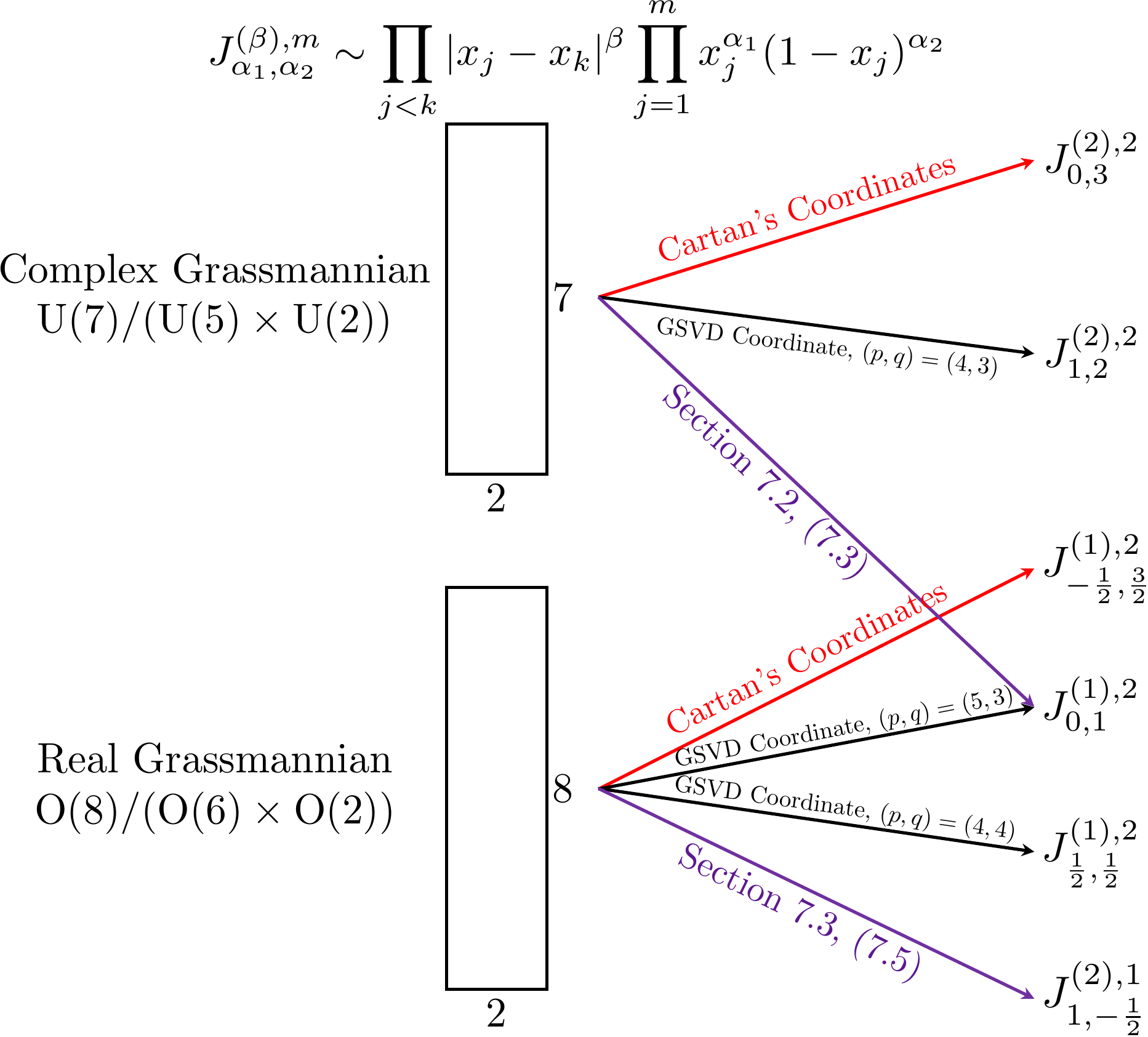}}
    \caption{Examples illustrating the lack of a one-to-one relationship between symmetric spaces and classical random matrix theories: A complex Grassmannian (top) obtains three Jacobi ensembles. A real Grassmannian (bottom) obtains four Jacobi ensembles. Particularly, the $\beta=1$ Jacobi ensemble $J_{0,1}^{(1),2}$ can be obtained from both symmetric spaces. Interestingly, a complex Grassmannian can lead to ({\color{violet}top purple}) a real RMT in the sense that $\beta=1$. Similarly a real Grassmannian obtains $\beta=2$ RMT ({\color{violet}bottom purple}).}
    \label{fig:manytomany}
\end{figure}

\par This work shows that  a symmetric space can be associated to multiple random matrix theories (Figure \ref{fig:manytomany}). Letting go of the arbitrariness of the choice of the KAK decomposition coordinate system allows us to choose other coordinate systems on symmetric spaces, thereby leading us to the complete list of classical random matrix ensembles (Sections \ref{sec:circular}, \ref{sec:jacobi}, \ref{sec:hermite}, \ref{sec:Laguerre}). Many of these coordinate systems are sometimes better known as matrix factorizations, used widely in matrix models of the classical ensembles \cite{dumitriu2002matrix,edelman2005random,edelman2008beta,forrester2010log,killip2004matrix}. However, in Section \ref{sec:mixed}, we compute new families of the Jacobi ensemble parameters from coordinate systems that have not been known before. 

\par This work also endeavors to make the Lie theory more widely accessible, by simplifying and modernizing key ideas and proofs in \cite{Helgason1984}. Cartan's theory \cite{Cartan1,cartan1927certaines,Cartan2,Cartan3} as developed by Helgason \cite{Helgason1978,Helgason1984} is a crowning mathematical achievement, and it is our hope to open up this theory for the benefit of all. Indeed, in \cite[p. 428]{segel2009recountings}, Helgason writes about the difficulty of understanding Cartan's writings:
\begin{quote}
    \small\textit{[Cartan] was one of the great mathematicians of the period, but his papers were quite a challenge. Hermann Weyl, in reviewing a book by Cartan from 1937 writes: ``Cartan is undoubtedly the greatest living master in differential geometry... I must admit that I found the book like most of Cartan's papers, hard reading."}
\end{quote}
In the same vein, while we are admirers of Helgason's extensive work, we authors must admit that we in turn found \cite{Helgason1978,Helgason1984} hard reading as well, and this paper attempts to introduce the theory by couching the ideas in terms of what we call ping pong operators. 

\par Summarizing our work:
\begin{itemize}
    \item We use the coordinate systems of the $\kak$ decomposition which connects a single symmetric space to multiple random matrices (Figure \ref{fig:manytomany}), completing the list of associated classical random matrix ensembles.
    \item  We translate some of the key concepts in Cartan's theory of symmetric spaces into easier to follow linear algebra (Section \ref{sec:LinAlg}). 
    \item  We provide coordinate systems (matrix factorizations) of symmetric spaces that have not been discussed in random matrix context, obtaining new parameter families of Jacobi ensemble (Section \ref{sec:mixed}).
\end{itemize}

%This work shows that a one to one correspondence between the classical random matrix ensembles and symmetric spaces, as suggested by many authors, leads to an irreconcilable situation. If one starts with symmetric spaces, one has to make what we call a coordinate system choice, what others might call a matrix factorization choice. This choice has been the map $\Phi : K\times A\to G/K; (k, a)\mapsto kaK$ of Cartan, which we could call the KAK decomposition. The result is that not all Jacobi ensembles are obtained. 

%\par Letting go of the arbitrariness of the choice of the KAK decomposition allows us to choose other coordinate systems on symmetric spaces, thereby leading us to complete the list of classical random matrix ensembles. These coordinate systems could be identified as matrix factorizations and thus used as sampling methods. 

%\par We work out and simplify the connections to other areas of mathematics, in particular the theory of double cosets. As a result we conclude that the emerging list of random matrix theories that arise is slightly larger than the list of classical random matrix theories. We also suggest new schemes for sampling random matrix ensembles.

% in a way that the authors hope may be more insightful for a much larger audience. 

\section{Background}\label{sec:background}
%Our interest in Lie theory emanates from the fact that the famous $K$'s, $P$'s, and the refinement of the $P$'s to the $A$'s generate so many of the important factorizations of applied mathematics. With symmetric spaces we only have Jacobians as products of sines of the ``restricted roots," but random matrix theory requires cosines in the case of Jacobi, which is where the generalized Cartan decomposition comes in.

\subsection{Joint densities of classical random matrix ensembles} 
\par Dyson introduced the $\beta=1,2,4$ circular ensembles \cite{dyson1962a,dyson1962threefold} in 1962. Earlier expositions on circular ensembles could be found on Hurwitz \cite{hurwitz1963ueber}, and Weyl \cite{weyl1946classical}. Hermite ensembles were introduced by Wigner\footnote{\cite{mehta1960statistical} credits \cite{hsu1939distribution} for the GOE. In fact \cite{hsu1939distribution} has the Jacobian for the symmetric real eigenvalue problem, and indeed works with $AA^T$ where \verb|A=randn(m,n)| but does not work with $A+A^T$. No doubt \cite{hsu1939distribution} could have instantly written down the GOE distribution if he had only been asked.} \cite{wigner1955characteristic,wigner1958distribution}. Laguerre and Jacobi ensembles could be found as early as 1939 in the statistics literature by Fisher \cite{fisher1939sampling}, Roy \cite{roy1939p} or Hsu \cite{hsu1939distribution}. The physics literature first touches upon the idea of Laguerre and Jacobi with the 1963 thesis of Leff \cite{leffthesis}. The following list is the joint probability densities (without normalization constants) of classical random matrix ensembles $(\beta=1,2,4)$:
\begin{itemize}\small
    \item Circular : $\displaystyle\prod_{j<k}|e^{i\theta_j} - e^{i\theta_k}|^\beta$, \hspace{1cm} $(\theta_1, \dots, \theta_n)\in[0, 2\pi)^n$ \vspace{-0.1cm}
    \item Hermite : $\displaystyle\prod_{j<k}|\lambda_j - \lambda_k|^\beta e^{-\sum\frac{\lambda^2_j}{2}}$ \hspace{1cm} $(\lambda_1, \dots, \lambda_n)\in\mathbb{R}^n$\vspace{-0.1cm}
    \item Laguerre : $\displaystyle\prod_{j<k}|\lambda_j - \lambda_k|^\beta\displaystyle\prod_{j=1}^m\lambda_j^\alpha e^{-\sum\frac{\lambda_j}{2}}$\hspace{1cm} $(\lambda_1, \dots, \lambda_m)\in[0, \infty)^m$\vspace{-0.1cm}
    \item Jacobi : $\displaystyle\prod_{j<k}|x_j - x_k|^\beta\displaystyle\prod_{j=1}^m x_j^{\alpha_1}(1-x_j)^{\alpha_2}$\hspace{1cm} $(x_1, \dots, x_m)\in[0, 1]^m$
\end{itemize}
In particular the parameters $\alpha, \alpha_1, \alpha_2 > -1$ are quantized as integer multiples of $\frac{\beta}{2}$, i.e., $\frac{\beta}{2}(N+1) - 1$ for some nonnegative integer $N$. 

\subsection{Symmetric space and the generalized Cartan decomposition}\label{sec:gencartanbackground}

\par In this section we introduce the theory related to the generalized Cartan decomposition. For readers without preliminary knowledge in Lie theory, we recommend skipping to Section \ref{sec:LinAlg} which follows a more modern linear algebra approach.

\par Let $G/K_\sigma$ be a Riemannian symmetric space with a real reductive noncompact Lie group $G$ and its maximal compact subgroup $K_\sigma$. Let $\sigma$ be the Cartan involution on $\mathfrak{g}:=\text{Lie}(G)$. Then $\mathfrak{g} = \mathfrak{k}_\sigma+\mathfrak{p}_\sigma$ is the Cartan decomposition. Let $\tau$ be another involution on $\mathfrak{g}$ such that $\tau\sigma = \sigma\tau$ and let $\mathfrak{g} = \mathfrak{k}_\tau+\mathfrak{p}_\tau$ be the $\pm 1$ eigenspace decomposition by $\tau$. Denote by $K_\tau$ the analytic subgroup of $G$ with tangent space $\mathfrak{k}_\tau$. Let $\mathfrak{a}$ be a maximal abelian subalgebra of $\mathfrak{p}_\tau\cap\mathfrak{p}_\sigma$ and define $A := \exp(\mathfrak{a})$. We introduce the (noncompact) \textit{generalized Cartan decomposition} \cite[Theorem 4.1]{flensted1978spherical}.
\begin{theorem}[generalized Cartan decomposition, $\kak$ decomposition]\label{thm:noncompactk1ak2decomp}
With the above setting, we have the following decomposition of $G$:
\begin{equation}
G = K_\tau A K_\sigma.
\end{equation}
That is, for any $g\in G$, we have $k_1\in K_\tau, k_2\in K_\sigma$ and $a \in A$ such that $g = k_1 a k_2$. 
\end{theorem}

We often use the equivalent name ``$\kak$ decomposition" for simplicity. Note that if $\tau=\sigma$ (i.e., $K = K_\sigma = K_\tau$) we recover the usual KAK decomposition, $G = K A K$. The generalized Cartan decomposition in Flensted-Jensen's work \cite{flensted1978spherical} is originally intended for the case where $G$ is noncompact. The compact analogue is developed by Hoogenboom \cite[Theorem 3.6]{hoogenboom1983generalized}. 

\begin{theorem}[generalized Cartan decomposition, compact case]\label{thm:compactk1ak2decomp}
Let $G/K_\sigma$ and $G/K_\tau$ be two compact Riemannian symmetric spaces. Let $\mathfrak{g} = \mathfrak{k}_\sigma+\mathfrak{p}_\sigma$ and $\mathfrak{g} = \mathfrak{k}_\tau+\mathfrak{p}_\tau$ be the corresponding eigenspace decompositions of $\mathfrak{g}=\lie(G)$. Then, for a maximal abelian subalgebra $\mathfrak{a}$ of $\mathfrak{p}_\sigma\cap\mathfrak{p}_\tau$ and $A = \exp(\mathfrak{a})$ we have
\begin{equation*}
    G = K_\tau A K_\sigma.
\end{equation*}
\end{theorem}

From the space of linear functionals $\mathfrak{a}^*$, we collect eigenvalues of an adjoint representation (the commutator) of $\mathfrak{a}$ on $\mathfrak{g}$ and call these eigenvalues the roots of the $\kak$ decomposition. By fixing the  Weyl chamber, we obtain a set of positive roots $\Sigma^+$. Details of the theory of the $\kak$ decomposition and its root system can be found in Flensted-Jensen \cite{flensted1978spherical,flensted1980discrete}, Hoogenboom \cite{hoogenboom1983generalized}, Matsuki \cite{matsuki1995double,matsuki1997double,matsuki2002classification} and Kobayashi \cite{kobayashi2007generalized}. The $\kak$ decomposition is also studied in the context of spherical harmonics and intertwining functions \cite{hoogenboom1983intertwining,james1974generalized}. Refine the root space $\mathfrak{g}_\alpha$ of a root $\alpha$ by $\pm 1$ eigenspaces of $\sigma \tau$. Let the two dimensions be $m_{\alpha}^{\pm}$. 

Let $dk_\sigma$, $dk_\tau$ be the Haar measures of $K_\sigma$, $K_\tau$, respectively. Let $dH$ be the Euclidean measure on $\mathfrak{a}$. The Jacobian of the $\kak$ decomposition is the following.

\begin{theorem}[Jacobian of the $\kak$ decomposition \cite{flensted1980discrete,hoogenboom1983generalized}]\label{thm:k1ak2jacobians} 
Let $dg$ be the Haar measure on $G$ and let $H\in\mathfrak{a}$. We have the Jacobian and the integral formula corresponding to the change of variables associated with the $\kak$ decomposition,
\begin{equation*}
    \int_G f(g)dg = \int_{K_\tau}\int_{K_\sigma}\int_{\mathfrak{a}^+} f(k_\sigma \exp(H) k_\tau) d\mu(H)dk_\sigma dk_\tau,
\end{equation*}
where for noncompact $G$
\begin{equation}\label{eq:noncompactk1ak2jac}
    d\mu(H) \propto \prod_{\alpha\in\Sigma^+}(\sinh\alpha(H))^{m_{\alpha}^+} (\cosh\alpha(H))^{m_{\alpha}^-}dH,
\end{equation}
and for compact $G$
\begin{equation}\label{eq:compactk1ak2jac}
    d\mu(H) \propto \prod_{\alpha\in\Sigma^+}(\sin\alpha(H))^{m_{\alpha}^+} (\cos\alpha(H))^{m_{\alpha}^-}dH.
\end{equation}
\end{theorem}

\par Similar results on the KAK decomposition and the restricted roots of symmetric spaces can be found in standard Lie group textbooks \cite{bump2004lie,gilmore2012lie,Helgason1978,Helgason1984,knapp2013lie}. In the KAK case, the Jacobian \eqref{eq:noncompactk1ak2jac} reduces down to $\prod{(\sinh\alpha(H))^{m_\alpha}}$ as we do not have $-1$ eigenspace of $\sigma\tau$ so that $m_\alpha = m_\alpha^+$ \cite{Helgason1978,kirillov1995representation,knapp2001representation}. 

\par Theorems \ref{thm:noncompactk1ak2decomp}, \ref{thm:compactk1ak2decomp}, and \ref{thm:k1ak2jacobians} are decompositions of the group $G$. These decompositions can also be applied to the symmetric space $G/K_\sigma$. The following map $\Phi$ is the $\kak$ decomposition of the Riemannian symmetric space $G/K_\sigma$. The map $\Phi$ is also called the Hermann action \cite{hermann1960variational,kollross2002classification}, nonstandard polar coordinates \cite{zirnbauer1995single}, non-Cartan parametrization \cite{caselle2004random}. In the KAK case $(K = K_\sigma=K_\tau)$, Helgason called this the polar coordinate decomposition \cite{Helgason1978} and credits Cartan \cite{cartan1927certaines} for this map. Since the $G$-invariant measure of $G/K$ inherits the Haar measure of $G$, the identical Jacobian is obtained for the decomposition of a symmetric space.\footnote{The actual development of the Jacobians \eqref{eq:noncompactk1ak2jac}, \eqref{eq:compactk1ak2jac} was done the other way around. In \cite{Helgason1962}, Helgason credits Cartan \cite{Cartan3} for the derivation of these Jacobians which was then computed only for symmetric spaces. The KAK decomposition was discovered later in 1950s and the Jacobians are identically extended from the decomposition of $G/K$ to the decomposition of $G$.}
\begin{theorem}[$\kak$ decomposition of $G/K_\sigma$]\label{thm:symspacejacobians}
Given a $\kak$ decomposition $G = K_\sigma A K_\tau$ with the Riemannian symmetric space $G/K_\sigma$, we have the map $\Phi$,
\begin{equation}\label{eq:symspacekaK}
    \Phi:K_\tau\times A\to G/K;\hspace{0.2cm} (k_\tau, a)\mapsto k_\tau aK.
\end{equation}
Suppose $H\in\mathfrak{a}$, $a = \exp(H)$. For the $G$-invariant measure $dx$ of $G/K_\sigma$, $dk_\tau = \text{Haar}(K_\tau)$ and the Euclidean measure $dH$ on $\mathfrak{a}$, $dx = dk_\tau d\mu(H)$ holds where the Jacobian $d\mu(H)$ is given in \eqref{eq:noncompactk1ak2jac} if $G$ is noncompact and \eqref{eq:compactk1ak2jac} if $G$ is compact. 
\end{theorem}

\begin{remark}[Representing $G/K\cong P$: $gK$ (coset) or $p\in P$?]\label{rem:representations}
In the standard KAK decomposition, the Jacobian \eqref{eq:noncompactk1ak2jac} (resp. \eqref{eq:compactk1ak2jac}) only has $\sinh$ (resp. $\sin$) terms as we discussed above. This result could be found in many literature, where some authors \cite{flensted1980discrete,Helgason1962,Helgason1984,knapp2001representation} use $\prod\sinh\alpha(H)$ as the Jacobian, whereas other authors \cite{duenez2004random,kirillov1995representation,terras2016harmonic} use $\prod\sinh(\alpha(H)/2)$. This gap is due to the difference in the realization of a symmetric space $G/K$ as a subset $P\subset G$. The former uses the right coset representative, i.e., $G/K\to P$ as $gK\mapsto p$ where $g=pk$ is its group level Cartan decomposition. Then the action of $G$ on $G/K$ is given as $(g_1, g_2K)\mapsto g_1g_2K$. The latter authors use the map $G/K\to P$ such that $gK\mapsto g(\sigma g)^{-1}$ where $\sigma$ is the group level involution. The $G$-action is $(g_1, g_2)\mapsto g_1g_2(\sigma g_1)^{-1},\,\,g_1\in G, g_2\in P$. In terms of Theorem \ref{thm:symspacejacobians}, the latter gives the map $\Phi$ such that $(k, a)\mapsto ka^2k^{-1}$ since 
\begin{equation*}
    g(\sigma g)^{-1} = pk\sigma(pk)^{-1} = pk(p^{-1}k)^{-1} = pkk^{-1}p = p^2 = kak^{-1}kak^{-1} = ka^2k^{-1}
\end{equation*}
which explains the extra factor $\frac{1}{2}$ applied to $H$ where $a = \exp(H)$. Moreover, these two identifications define the map $\Phi:K\times A\to P$ with the same $k, a$ as 
\begin{equation}\label{eq:mapPhi}
    \Phi:(k, a)\mapsto kaK\hspace{0.5cm}\text{or}\hspace{0.5cm}\Phi:(k, a)\mapsto ka^2k^{-1},
\end{equation}
\noindent 
depending on the author's notational choice explained above. This coordinate system $\Phi$ is sometimes called the polar coordinate decomposition, e.g., see Helgason \cite[p.402]{Helgason1978}.
\end{remark}

\begin{example}[$G/K$ vs. $P$: A symmetric positive definite matrix]
Let us take a look at the two realizations in Remark \ref{rem:representations}, for $G/K = \gl{n, \mathbb{R}}/\ortho{n}$, where $P$ is the set of all symmetric positive definite matrices. Let $S$ be a fixed positive definite symmetric matrix, with its eigendecomposition $S = Q\Lambda Q^T$, with $Q\in\ortho{n}$. The coset representation of $S$ is $Q\Lambda\cdot\ortho{n}\in G/K$ as $Q\Lambda = (Q\Lambda Q^T)Q$ is the polar decomposition. With the realization of $P\cong G/K$, the point in $G/K$ is represented by the matrix $S = Q\Lambda Q^T$.
\end{example}

Finally we have the Lie algebra counterpart of Theorem \ref{thm:symspacejacobians} when $K = K_\sigma=K_\tau$. 
\begin{theorem}\label{thm:liealgebrajac}
For a noncompact Riemannian symmetric space $G/K$ with the Cartan decomposition $\mathfrak{g} = \mathfrak{k}+\mathfrak{p}$ let $\mathfrak{a}$ be a maximal abelian subalgebra of $\mathfrak{p}$. We have 
\begin{equation}\label{eq:liealgebramap}
    \Psi:K\times\mathfrak{a}\to \mathfrak{p};\hspace{0.2cm} (k, H)\mapsto kHk^{-1},
\end{equation}
equivalently the decomposition $\mathfrak{p} = \cup_{k\in K} k \mathfrak{a} k^{-1}$ with the Jacobian $d\mu$ given as 
\begin{equation}\label{eq:noncompactliealgebrajac}
    d\mu(H) \propto \prod_{\alpha\in\Sigma^+} |\alpha(H)|^{m_\alpha},
\end{equation}
where $H\in\mathfrak{a}$ and $\Sigma$ is the restricted root system with dimensions $m_\alpha$. The measure on $\mathfrak{p}$ is the Euclidean measure. 
\end{theorem}

\subsection{A symmetric space: one RMT or many RMTs?}\label{sec:onevsmany}

\par The answer to the title question of this section is that both one and many can be construed as correct. To explain how this is possible requires teasing apart the assumptions behind the words ``associated with." Certainly, \cite{altland1997nonstandard,caselle1996new,duenez2004random,ivanov2002random} associate one random matrix with one symmetric space. However the example of the GSVD coordinate systems discussed in Section \ref{sec:introGSVD} associates multiple Jacobi densities with one symmetric space, the Grassmannian manifold. In \cite{caselle2004random} another example is illustrated as the ``non-Cartan parameterization," for the special case of $(G, K_\sigma, K_\tau) = (\un{n}, \ortho{n}, \un{p}\times \un{q})$. (A similar approach may be found in \cite{an2006generalization}.) This is discussed in Section \ref{sec:newjacobi1}.

%$J_{0, 3}^{2, 2}, J_{2, 1}^{2, 2}, J_{0, 1}^{1, 2}$
%$J_{-\frac{1}{2}, \frac{3}{2}}^{1,2}, J_{0,1}^{1,2}, J_{\frac{1}{2}, \frac{1}{2}}^{1,2}, J_{1, -\frac{1}{2}}^{1,1}$

\par The reconciliation is that indeed it is true that the required maps \eqref{eq:symspacekaK} when $K = K_\sigma = K_\tau$ i.e., $\Phi(k, a) = kaK = kak^{-1}$ (compact) or the map \eqref{eq:liealgebramap} $\psi(k, H) = kHk^{-1}$ (noncompact) leads to a unique random matrix theory associated to a given symmetric space $G/K$. This is unique in a sense that any geodesic on the symmetric space $G/K$ could be transformed to the geodesic on $A$ with the above maps. 

\par However, if we relax the condition so that we are allowed to choose $K_\tau$ under the generalized Cartan decomposition framework, we can associate multiple random matrix theories to one symmetric space. The GSVD coordinate systems in Section \ref{sec:introGSVD} illustrate this viewpoint. The real Grassmannian manifold $G/K = \ortho{n}/(\ortho{n-s}\times \ortho{s})$ has the map $\Phi:(k, a)\mapsto kaK$ for $K=K_\sigma=K_\tau$ explicitly written as $X = \text{\scalebox{0.75}[0.75]{$\begin{bmatrix}U & \\ & V\end{bmatrix}\begin{bmatrix}C \\ S\end{bmatrix}$}}\cdot\ortho{s}$ where $U, V$ are $(r-s)\times s, s\times s$ orthogonal matrices. On the other hand, if we let $K_\tau = \ortho{p}\times \ortho{q}$, we have multiple maps $\Phi:(k_\tau, a)\mapsto k_\tau a K$ written as $X = \text{\scalebox{0.75}[0.75]{$\begin{bmatrix}U & \\ & V\end{bmatrix}\begin{bmatrix}C \\ S\end{bmatrix}$}}\cdot\ortho{s}$ where $U, V$ are $p\times s, q\times s$ orthogonal matrices. 

\par Starting from Section \ref{sec:circular}, we discuss \rn{1} random matrices arising from the $\kak$ decompositions of compact symmetric spaces (Theorem \ref{thm:symspacejacobians} or \ref{thm:compactk1ak2decomp}) and \rn{2} random matrices arising from the Lie algebra decomposition of noncompact symmetric spaces (Theorem \ref{thm:liealgebrajac}). The associated decompositions are well explained by matrix factorizations in numerical linear algebra. As we pointed out, the resulting Jacobi ensembles cover the full parameter set of the classical Jacobi densities, thereby completing the classification from the classical RMT point of view.

\section{Cartan's idea: a modernized approach }\label{sec:LinAlg}

The Jacobian of the KAK ($\kak$) decomposition, equivalently the determinant of the differential of the map $\Phi:K\times A \to P$ (in Theorem \ref{thm:symspacejacobians} and Remark \ref{rem:representations}), is computed in several references \cite{Helgason1984,kirillov1995representation,knapp2001representation}. The proof of \eqref{eq:noncompactk1ak2jac} is can also be found in \cite{flensted1980discrete,hoogenboom1983generalized}. However the proof can be inaccessible to some audiences. Meanwhile, individual cases of the KAK decomposition, recognized as matrix factorizations, show up in many areas of mathematics, and some were discovered in various formats by specialists in numerical linear algebra. Motivated by Random Matrix Theory (and sometimes perturbation theory in numerical analysis), Jacobians of these factorizations were often computed case-by-case using the matrix differentials and wedging of independent elements \cite{dumitriu2002matrix,edelman2018random,forrester2010log,mehta2004random}.

\par In this section, we provide a generalization of such individual Jacobian computations and compare it to the general technique Helgason proposed. With appropriate translation of terminologies and maps in Lie theory into linear algebra, we observe both methods are indeed the same process but have been illustrated in different languages for a long time. We start out by introducing some important concepts in Lie theory accessible to an audience with a good background in linear algebra and perhaps some basic geometry. Then, in Table \ref{tab:comparison1}, we present a line-by-line correspondence between Helgason's derivation and the proof by matrix differentials.

\subsection{The ping pong operator, ping pong vectors and ping pong subspaces}\label{sec:pingpong}
We will start with a concrete $2\times 2$ linear operator so as to establish the notions of the \textit{ping pong operator}, \textit{ping pong vectors}, \textit{ping pong
subspaces} and the relationship to eigenvectors. Then we will define a ``bigger" linear operator $\ad_H$ that acts on $2\times 2$ spaces exactly in the manner we are about to describe.
\par We introduce the $2\times 2$ matrix 
\begin{equation*}
    M:=\twotwo{0}{\alpha}{\alpha}{0} = \alpha\twotwo{0}{1}{1}{0},
\end{equation*}
which we will call a \textit{2 by 2 ping pong operator} and we will call $\twoone{1}{0}$ and $\twoone{0}{1}$ the \textit{ping pong vectors} of $M$, in that $M$ bounces these two vectors into $\alpha$ times the other,
\begin{equation*}
M\twoone{1}{0} = \alpha\twoone{0}{1}, \hspace{0.5cm}M\twoone{0}{1} = \alpha\twoone{1}{0}.
\end{equation*}
Furthermore $M$ has eigenvectors $\twoone{1}{1}, \twoone{\ 1}{-1}$, with eigenvalues $\alpha, -\alpha$. We will call the eigenvalue a \textit{root} of $M$.

\par Also worth pointing out are the matrix exponential and matrix $\sinh$ of $M$,
\begin{equation*}
e^M = \twotwo{\cosh\alpha}{\sinh\alpha}{\sinh\alpha}{\cosh\alpha}\hspace{0.5cm}\text{ and }\hspace{0.5cm} \sinh M =
\frac{1}{2}(e^M + e^{-M})=
\sinh\alpha \cdot \twotwo{0}{1}{1}{0},  
\end{equation*}
and thus we see that $\sinh M$ is another ping pong operator with scaling $\sinh\alpha$. Figure \ref{fig:kpvecs} plots the action of a ping pong matrix and its exponential, with notations that we will use in the next sections, i.e, the ping pong operator is denoted $\ad_H$, $p_j$ and $k_j$ are the ping pong vectors, and $x_j$ and $\theta x_j$ are the eigenvectors. The right side of Figure \ref{fig:kpvecs} shows the action of $e^M$ and portrays $\sinh(M)$ as a projection of $e^M$ on the $p_j$ direction.

\par We now go beyond $2\times 2$ matrices, and suggest the more general $2N \times 2N$ ping pong matrix $M_N$, with $N$ roots, $\alpha_1,\ldots,\alpha_N$, $N$ pairs of ping pong vectors $(k_1, p_1), \dots, (k_N, p_N)$ along with eigenvectors $(x_1, y_1), \dots, (x_N, y_N)$,

\begin{equation*}
    M_N = \begin{bsmallmatrix}
    0 & \alpha_1 & & & \\
    \alpha_1 & 0 & & & \\
    & &\ddots & &\\
    & & & 0 & \alpha_N\\
    & & & \alpha_N & 0
    \end{bsmallmatrix}
\end{equation*}

\begin{equation*}
    k_j, p_j, x_j, y_j = \begin{bmatrix}
    \vdots\\ 1 \\ 0\\ \vdots
    \end{bmatrix}, \begin{bmatrix}
    \vdots\\ 0 \\ 1 \\ \vdots
    \end{bmatrix}, \begin{bmatrix}
    \vdots\\ 1 \\ 1\\ \vdots
    \end{bmatrix}, \begin{bmatrix}
    \vdots\\ {\ }1 \\ -1 \\ \vdots
    \end{bmatrix}, \hspace{0.5cm} j = 1, 2, \dots, N,
\end{equation*}
where the $2j-1$ and $2j$ positions are $0$ or $\pm 1$ and all other entries of these vectors are $0$. The matrices $\exp(M_N)$ and $\sinh M_N$ are block versions of the $2 \times 2$ case. 

\par We may define the subspaces, $\mathfrak{k}$ and $\mathfrak{p}$ (using the ``mathfrak" Fraktur letters ``k" and ``p") to be the span of the $k_j$ and $p_j$ respectively. Notice that $\mathfrak{k}$ and $\mathfrak{p}$ are orthogonal complements as subspaces. A key ``ping pong" relationship between these subspaces is that
\begin{gather*}
M_N k \in\mathfrak{p}  \text{\  if $k\in\mathfrak{k}$},\\
M_N p \in\mathfrak{k} \text{\  if $p\in\mathfrak{p}$}.
\end{gather*}

Thus, if we consider ${M_N}{|_\mathfrak{k}}$, the restriction of $M_N$ to $\mathfrak{k}$ we have an operator from $\mathfrak{k}$ to $\mathfrak{p}$. Evidently, ${M_N}{|_\mathfrak{k}}$ as a matrix may be obtained by taking the even rows and odd columns of $M_N$. The result is a diagonal matrix with the $\alpha_j$ on the diagonal. Similarly $\sinh(M_N){|_\mathfrak{k}}$ is a diagonal matrix with $\sinh(\alpha_j)$ on the diagonal. We then get the important result that
\begin{equation*}
\det(\sinh(M_N){|_\mathfrak{k}}) = \prod_{j=1}^N \sinh \alpha_j,
\end{equation*}
the product of the hyperbolic sines of the roots.

\par Given a linear operator $\mathcal{L}$ on a vector space with nonzero eigenvalues $\pm\lambda$, the following lemma constructs a pair of ping pong vectors  from $\mathcal{L}$. 
\begin{lemma}\label{lem:linop}
For a linear operator $\mathcal{L}$ defined on any vector space, assume $\pm\lambda$ are both nonzero eigenvalues of $\mathcal{L}$. Let $x$ and $y$ be the corresponding eigenvectors, i.e., $\mathcal{L}x = \lambda x$ and $\mathcal{L}y = -\lambda y$. Define two vectors $k:=x+y, p:=x-y$. Then, $k, p$ are ping pong vectors. Furthermore we have for the operator $\exp(\mathcal{L})$,
\begin{equation*}
    e^\mathcal{L} k = \cosh\lambda k + \sinh\lambda p, \hspace{0.5cm}e^\mathcal{L} p = \sinh\lambda k + \cosh\lambda p.
\end{equation*}
\end{lemma}
\noindent The proof is a straightforward extension of the discussion in previous paragraphs.

\begin{remark}\label{rem:pingpongroles}
For the reader who wants to know the upcoming significance of this fact for Jacobians of matrix factorizations, it turns out (or maybe as the reader already observed in Section \ref{sec:background}) that the Jacobian will be the product of $\sinh\alpha$'s. Just as the matrix $\sinh\big(\begin{bsmallmatrix} 0 & \alpha \\ \alpha & 0\end{bsmallmatrix}\big)$ takes one of the ping pong vectors to $\sinh\alpha$ times the other, the key piece of the differential map will consist of multiple ping pong relationships, each one sending one ping pong vector to another.
\end{remark}

\subsection{The Kronecker product, linear operator $\ad_X$ and its exponential}
Lie theory picks out operators $\mathcal{L}$ that exactly have the properties in Section \ref{sec:pingpong}. Our vector spaces are now matrix spaces, and our operators are linear operators on a matrix space. We introduce the Lie bracket, denoted by $[X, Y]$, defined as $[X, Y] = XY - YX$(the commutator). The Kronecker product notation is very helpful in this context. We define the Kronecker product notation\footnote{Many authors would write $\text{vec}(BXA^T) = (A\otimes B)\text{vec}(X)$, but we omit the ``vec" as be believe it is always clear from context. In a computer language such as Julia, one would write \texttt{kron(A,B) * vec(X) = vec(B*X*A')}} as a linear operator on a matrix space. 
\begin{equation}\label{eq:kroneckerid}
    (A\otimes B)X = BXA^T.
\end{equation}
With this, we can express the Lie bracket with Kronecker products.
\begin{equation*}
    (I\otimes X - X^T\otimes I)Y = XY - YX.
\end{equation*}
Consider the Lie bracket as a linear operator (determined by $X$) applied to $Y$, and call this operator $\ad_X$: (abbreviation for ``adjoint") 
\begin{gather*}
    \ad_X =   I\otimes X - X^T\otimes I\\
    \ad_X(Y) = [X,Y].
\end{gather*}
This will be the important ping pong operator $\mathcal{L}$. The operator exponential of $\ad_X$ (equivalently, the matrix exponential of $I\otimes X - X^T\otimes I$) is given in the following.
\begin{lemma}\label{lem:expad}
For the linear operator $\ad_X$, the following holds for $e^{\ad_X}:=\sum_{j=0}^\infty \frac{(\ad_X)^n}{n!}$ and $\sinh \ad_x = (e^{\ad_X} + e^{-\ad_X})/2$:
\begin{gather}
    e^{\ad_X} = \exp (I\otimes X - X^T\otimes I)  = (e^{-X})^T\otimes e^X, \\
    e^{\ad_X}Y = e^X Y e^{-X} \hspace{0.5cm}\mbox{\rm and }\hspace{0.5cm}(\sinh \ad_x)Y = (e^X Y e^{-X} -  e^{-X} Y e^{X})/2.\label{eq:expad}   
\end{gather}
\end{lemma}
\begin{proof}
The proof is straightforward by the identity \eqref{eq:kroneckerid}. $e^XYe^{-X} = \big((e^{-X})^T\otimes e^X\big) Y$ and $e^{\ad_X}Y = \exp(I\otimes X - X^T\otimes I)Y$. It is left to prove $(e^{-X})^T\otimes e^X= \exp(I\otimes X - X^T\otimes I)$. Since $I\otimes X$ commutes with $ X^T\otimes I$, we have 
\begin{gather*} \exp(I\otimes X - X^T\otimes I) = e^{I\otimes X}e^{- X^T\otimes I} = (I \otimes e^X)((e^{-X})^T \otimes I) = (e^{-X})^T\otimes e^X,
\end{gather*}
proving the result. The $\sinh$ result follows trivially.
\end{proof}

\subsection{Antisymmetric and symmetric matrices: an important first example of symmetric space as ping pong spaces}\label{sec:examplesymasym}

In our first example, our vector space is $n \times n $ real matrices. Consider
\begin{gather*}
    \mathfrak{k} = \{\text{Antisymmetric matrices}\}\\
    \mathfrak{p} = \{\text{Symmetric matrices}\}.
\end{gather*}
The ping pong operator that will bounce $\mathfrak{k}$ and $\mathfrak{p}$ around will be $\ad_H= I \otimes H - H^T \otimes I,$ where $H$ is the diagonal matrix 
\begin{equation*}
H = \begin{bmatrix} h_1 & & \\  &\ddots& \\ & & h_n \end{bmatrix}.
\end{equation*}
Notice that the operator $\ad_H$ sends an antisymmetric matrix to a symmetric matrix, and a symmetric matrix to an antisymmetric matrix. 

\par What does this have to do with Jacobians of matrix factorizations such as the symmetric positive definite eigenvalue factorization? Consider a perturbation of $Q$ when forming $S=Q \Lambda Q^T$. An infinitesimal antisymmetric perturbation $Q^TdQ$ is mapped into a $dS$, an infinitesimal symmetric perturbation. This is the very linear map from the tangent space of $Q$ to that of $S$ that we wish to understand, so perhaps it is not surprising we would want to restrict our ping pong operator from $\mathfrak{k}$ to $\mathfrak{p}$. We invite the reader to check that the corresponding eigenmatrices and ping pong matrices of $\ad_H$ may be found in the first column of Table \ref{tab:kpexamples}.

\subsection{General $\mathfrak{k}$ and $\mathfrak{p}$ arise from an involution $\theta$}\label{sec:kpmatrices}
We proceed to construct more important general operators $\mathcal{L}$ that have the property in the assumption of Lemma \ref{lem:linop}. This is where the theory of Lie groups and symmetric spaces need to be brought in. Upon doing so, we will obtain two linear spaces of matrices $\mathfrak{k}$, $\mathfrak{p}$ and also a space $\mathfrak{a}$. 

\par For the reader not familiar with Lie groups, one need only imagine a continuous set of matrices which are a subgroup of real, complex, or quaternion matrices. The tangent space $\mathfrak{g}$ is just a vector space of matrix differentials at the identity. One key example is the compact Lie group $\ortho{n}$ (the group of square orthogonal matrices) and its tangent space at the identity $\mathfrak{g}_{\ortho{n}}$: the set of antisymmetric matrices. Another key example is all $n$-by-$n$ invertible matrices $\gl{n, \mathbb{R}}$ (a noncompact Lie group), and its tangent space $\mathfrak{g}_{\gl{n, \mathbb{R}}}$, consisting of all $n$-by-$n$ matrices. 

\par Cartan noticed that important matrix factorizations start with two ingredients: the \textbf{tangent space $\mathfrak{g}$} (at the identity) of a Lie group $G$ and an \textbf{involution $\theta$} on $\mathfrak{g}$. (i.e., $\theta^2 = \text{Id}$ and $\theta[X, Y] = [\theta X, \theta Y]$) An example of $\theta$ is $\theta(X)=-X^T$ on
$\mathfrak{g}$, for $G = \gl{n, \mathbb{R}}$. Among matrices in $\mathfrak{g}$, we select two kinds of matrices. The ones fixed by the involution $\theta$ and the ones negated by $\theta$. Denote each set by $\mathfrak{k}$ and $\mathfrak{p}$.
\begin{equation*}
    \mathfrak{k}:= \{g\in\mathfrak{g}:\theta(g) = g\}, \hspace{1cm}\mathfrak{p}:= \{g\in\mathfrak{g}:\theta(g) = -g\}.
\end{equation*}
(For $\gl{n, \mathbb{R}}$, these are the antisymmetric and symmetric matrices respectively.) 

\par The next important player is $\mathfrak{a}\subset\mathfrak{p}$. Readers familiar with the singular value decomposition know the special role of diagonal matrices in the SVD as they list the very important ``singular values". Diagonal matrices have the nice property that linear combinations are still diagonal, they commute (the Lie bracket of any two are zero), and they are symmetric (the $\mathfrak{p}$ of our first example). The generalization of this is to take a $\mathfrak{p}$, and find a maximal subalgebra where every matrix commutes. This is the maximal subspace $\mathfrak{a}\subset\mathfrak{p}$ such that for all $a_1, a_2\in \mathfrak{a}$, $[a_1, a_2] = 0$.

\par If $H \in \mathfrak{a}$, then $S=Q\Lambda Q^T$ is a symmetric positive definite eigendecomposition, with $\Lambda = e^H$. In the rest of the section we will be focusing on factorizations of the form $Q\Lambda Q^{-1}$ where $\Lambda$ is a matrix exponential of $H \in \mathfrak{a}$. (These will be more general than eigendecompositions, as $Q$ may not be orthogonal, and $\Lambda$ may not be diagonal.) In particular, we will compute the Jacobian of perturbations with respect to $Q$, holding $H$ constant, and thus necessarily the Jacobian will be defined in terms of $H$.

\par From here we assume that the Lie group $G$ is noncompact. The compact case will be discussed after completing the noncompact case. Pick $H\in\mathfrak{a}$, and recall that $\ad_H$ is a linear operator on $\mathfrak{g}$. The operator $\ad_H$ will play the role of $\mathcal{L}$, the ping pong operator. We decompose $\mathfrak{g}$ into the eigenspaces of $\ad_H$. For any eigenpair $(\alpha_j, x_j)$ of $\ad_H$, i.e., $\ad_H( x_j) = [H,  x_j]  = \alpha_j  x_j$, we observe (for $\alpha_j\neq 0$)
\[\ad_H(\theta x_j) = [H, \theta x_j] = -[-H, \theta x_j] = -[\theta H, \theta x_j] = -\theta([H, x_j]) = -\alpha_j \theta x_j,\] 
which implies the eigenvalues $\pm \alpha_j$ always exist in pairs, with corresponding eigenmatrices $x_j$ and $ \theta x_j$. This satisfies the assumption of  Lemma \ref{lem:linop}, from which we can now construct our ping pong matrices, 
\begin{equation}\label{eq:pingpongmatrix}
    k_j:=x_j + \theta x_j\hspace{1cm} p_j := x_j - \theta x_j,
\end{equation}
with the ping pong relationship by the operator $\ad_H$, 
\begin{equation}\label{eq:adHkp}
    \ad_H k_j = \alpha_j p_j \hspace{1cm} \ad_H p_j = \alpha_j k_j.
\end{equation}
Also the relationship by the operator $e^{\ad_H}$ follows,
\begin{gather}\label{eq:expadHkp1}
    e^{\ad_H} k_j = \cosh\alpha_j k_j + \sinh\alpha_j p_j,\\
    e^{\ad_H} p_j = \sinh\alpha_j k_j + \cosh\alpha_j p_j.\label{eq:expadHkp2}
\end{gather}
The ping pong matrices $k_j$, $p_j$, eigenmatrices $x_j$, $\theta x_j$ and the relationships \eqref{eq:adHkp}, \eqref{eq:expadHkp1} are illustrated in Figure \ref{fig:kpvecs}.  
\begin{figure}[h]
\fbox{\includegraphics[width=4.8in]{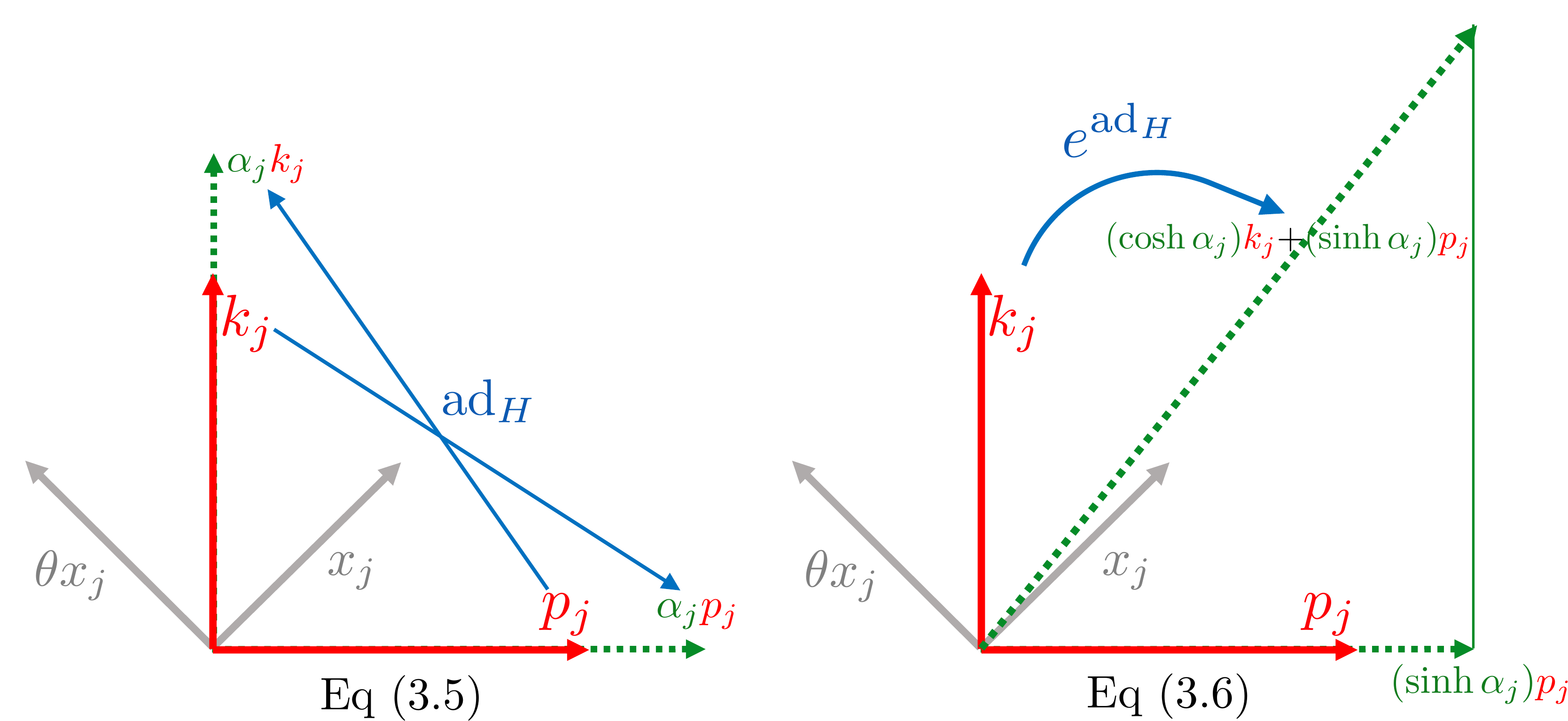}}
\caption{The eigenmatrices $x_j$, $\theta x_j$ and ping pong matrices $k_j, p_j$ \eqref{eq:pingpongmatrix} in the tangent space $\mathfrak{g}$. The operators are illustrated in blue lines. The operator $\ad_H$ and ping pong relationship (left), the operator $e^{\ad_H}$ on $k_j$ to $p_j$ (right). The left map shows the factor of $\alpha_j$, which is a building block of the Jacobian $\prod_j|\alpha_j(H)|$, \eqref{eq:noncompactliealgebrajac}. The factor of $\sinh\alpha_j$ in the right map builds the Jacobian $\prod_j|\sinh\alpha_j(H)|$, \eqref{eq:noncompactk1ak2jac}.}
\label{fig:kpvecs}
\end{figure}

\par As we mentioned in Remark \ref{rem:pingpongroles} and Section \ref{sec:examplesymasym}, the role of ping pong matrices $k_j, p_j$ are crucial. \textbf{The map $\boldsymbol{e^{\ad_H}}$ (particularly, $\boldsymbol{\sinh\ad_H}$) is the main ingredient constructing the differential map $\boldsymbol{d\Phi}$} of the factorization $\Phi:(Q, \Lambda)\mapsto Q\Lambda Q^{-1}$. The operator $e^{\ad_H}$ is applied to $k_j$ and then projected to the span of $p_j$ as in Figure \ref{fig:kpvecs} (right), leaving only the $\sinh\alpha_j$ factor. 

\par We now compute the full basis of $\mathfrak{k}$ and $\mathfrak{p}$. The collection $\cup_{j}\{x_j, \theta x_j\}$ is a full basis for the union of eigenspaces with nonzero eigenvalues. Since $\text{span}(\{x_j, \theta x_j\}) = \text{span}(\{k_j, p_j\})$ for any $j$, $\cup_j\{k_j, p_j\}$ is another full basis for the eigenspaces with nonzero eigenvalues. Interestingly, we observe $\theta k_j = k_j$ and $\theta p_j = -p_j$, which identifies $\cup_j\{k_j\}$ and $\cup_j\{p_j\}$ as subsets of the basis of $\mathfrak{k}$ and $\mathfrak{p}$ respectively. The remaining case is the zero eigenspace. When $\alpha_j=0$, there are two possibilities. Firstly, if $x_j$ and $\theta x_j$ are independent of each other, we can still obtain $k_j$ and $p_j$ as before and add them to $\cup_j\{k_j\}$ and $\cup_j\{p_j\}$. Secondly, if $x_j$ and $\theta x_j$ are colinear, $\theta x_j$ is either $x_j$ or $-x_j$. If $\theta x_j = x_j$ we collect such $x_j$ and name the set $K_z$. Similarly, if $\theta x_j = -x_j$ then we put them in $P_z$. Since we analyzed both nonzero and zero eigenspaces, we have obtained a full basis of $\mathfrak{g}$, which is $\big(\cup_j\{k_j, p_j\}\big)\cup K_z \cup P_z$. Refining once more, $\text{span}\big((\cup_j\{k_j\})\cup K_z\big) = \mathfrak{k}$ and $\text{span}\big((\cup_j\{p_j\})\cup P_z\big) = \mathfrak{p}$. 

\subsection{The operators ${\ad_H, e^{\ad_H}}$, and the subspaces ${\mathfrak{k}, \mathfrak{p}}$}
In Section \ref{sec:kpmatrices}, we obtained the basis of $\mathfrak{k}$ and $\mathfrak{p}$, in terms of ping pong matrices, by linearly combining eigenmatrices of the operator $\ad_H$. We now illustrate the relationship of the basis of $\mathfrak{k}$ and $\mathfrak{p}$ under $e^{\ad_H}$, just like we illustrated the operator $M_N$ in Section \ref{sec:pingpong}. In the $k_1, \dots, k_N$ and $p_1, \dots, p_N$ basis we have the following:
\begin{equation}\label{eq:expadHpingpong}
    e^{\ad_H}\begin{bmatrix}
    k_1\\p_1\\ \vdots \\ k_N\\p_N
    \end{bmatrix}
    = \begin{bmatrix}
    \cosh\alpha_1 & \sinh\alpha_1 & & &\\
    \sinh\alpha_1 & \cosh\alpha_1 & & &\\
    & & \ddots & &\\
    & & & \cosh\alpha_N & \sinh\alpha_N\\
    & & & \sinh\alpha_N & \cosh\alpha_N
    \end{bmatrix}\begin{bmatrix}
    k_1\\p_1\\ \vdots \\ k_N\\p_N
    \end{bmatrix}.
\end{equation}
We are now ready to carefully investigate the map $d\Phi$, using \eqref{eq:expadHpingpong}.

\begin{remark} Results in Lie theory imply that the eigenmatrices $x_j$ and $\theta x_j$ of $\ad_H$ are independent of the choice of $H\in\mathfrak{a}$. In other words, the complete basis of $\mathfrak{g}$ and $\mathfrak{k}$, $\mathfrak{p}$ obtained above does not care about a specific choice of $H$. Furthermore, the eigenvalues $\pm\alpha_j$ are functions of $H$ and these eigenvalue assigning functions $\tilde\alpha_j:H\mapsto\alpha_j\in\mathbb{R}$ are more properly called the \textit{restricted roots}. It can be inferred from the separation of the basis that $\mathfrak{k}$, $\mathfrak{p}$ together form the whole tangent space $\mathfrak{g}$. 
\begin{equation}\label{eq:cartan}
    \mathfrak{g} = \mathfrak{k}+ \mathfrak{p}.
\end{equation}
\end{remark}

\subsection{Symmetric spaces}
The reader may have noticed that our discussions have focused on the Lie algebras rather than the Lie groups themselves. It is a point of fact, that Lie groups are mostly useful to define the factorizations of our interest, but Lie algebras are where the Jacobian ``lives" and hence this is the most important place to concentrate. For the interested reader, the subgroup $K$ of $G$ is picked such that its tangent space is exactly $\mathfrak{k}$ (one easy way to imagine such a subgroup is to define $K:=\exp(\mathfrak{k})$), and we now obtain a \textit{symmetric space} $G/K$.

\par It can be proven that for the noncompact Lie group, there exists a unique involution $\theta$ such that the subgroup $K$ is the maximal compact subgroup of $G$. We call $\theta$ the \textit{Cartan involution} and \eqref{eq:cartan} is called the \textit{Cartan decomposition}. Furthermore the subset $P:=\exp(\mathfrak{p})$ plays an important role as its elements serve as representatives of the cosets in $G/K$. Regarding the identification of $G/K$ as elements in $P$, refer to the remark \ref{rem:representations}, where we point out as an example, taking $G/K = \gl{n, \mathbb{R}}/\ortho{n}$ that an element of $G/K$ has the form of a coset $gK$, then $gg^T$ may be a representative of the coset in $\mathfrak{p}$. While some authors use $(gg^T)^{1/2}$, the key point being each choice is well defined independent of choice of representative.

\subsection{When $G$ is a compact Lie group}
Upon considering the compact cases, it is helpful to make use of a certain duality between compact and noncompact symmetric spaces. We again start with a noncompact Lie group $G$ and the Cartan involution $\theta$. Let $\mathfrak{g} = \mathfrak{k}+\mathfrak{p}$ be the Cartan decomposition. Then, define a new space,
\begin{equation}
    \mathfrak{g}_C:= \mathfrak{k} + i\mathfrak{p},
\end{equation}
where $i$ is the imaginary unit. The result in Lie theory implies that the new vector space $\mathfrak{g}_C$ is the tangent space of a compact Lie group, say $G_C$. In table \ref{tab:kpexamples}, the first and third columns labeled $\gl{n, \mathbb{R}}/\ortho{n}$, $\ortho{p, q}/(\ortho{p}\times\ortho{q})$ are noncompact tangent spaces. Their compact duals are, respectively, the second and fourth columns labeled $\un{n}/\ortho{n}$, $\ortho{n}/(\ortho{p}\times\ortho{q})$.
\edef\savedbaselineskip{\the\baselineskip\relax}
\begin{table}[h]
    \scriptsize
    \centering
    \begin{tabular}{c|c|c|c|c}
    $\dfrac{G}{K}$ & $\dfrac{\gl{n, \mathbb{R}}}{\ortho{n}}$ & $\dfrac{\un{n}}{\ortho{n}}$ & $\dfrac{\ortho{p, q}}{\ortho{p}\times\ortho{q}}$ & $\dfrac{\ortho{n}}{\ortho{p}\times\ortho{q}}$\\
    \hline
    $x_l$ & 
    \begin{blockarray}{ccc}
    \begin{block}{ccc} & $j$ & $k$ \\ \end{block}
    \begin{block}{c[cc]}
       $j$ & 0 & 1\\ $k$ & 0 & 0 \\
    \end{block}
  \end{blockarray} & - &
    \begin{blockarray}{ccccc}
    & $j$ & $k$ & $j'$ & $k'$ \\
    \begin{block}{c[cc|cc]}
    $j$ & & $1$ & & $1$\\
    $k$ & $\mi1$ & & $1$ & \\ \cline{2-5}
    $j'$ & & $1$ & & $1$\\
    $k'$ & $1$ & & $\mi1$ &\\
    \end{block}
    \end{blockarray}  & - \\
    \hline
    $\theta x_l$ & \begin{blockarray}{ccc}
    \begin{block}{ccc} & $j$ & $k$ \\ \end{block}
    \begin{block}{c[cc]}
       $j$ & 0 & 0\\ $k$ & \mi 1 & 0 \\
    \end{block}
  \end{blockarray} & -
    & \begin{blockarray}{ccccc} & $j$ & $k$ & $j'$ & $k'$ \\
    \begin{block}{c[cc|cc]}
    $j$ & & 1 & & \mi 1 \\
    $k$ & \mi1 & & \mi1 & \\ \cline{2-5}
    $j'$ & & \mi1 & & 1\\
    $k'$ & \mi1 & & \mi1 &\\
    \end{block}
    \end{blockarray}  & - \\
    \hline
    $k_l$ & \begin{blockarray}{ccc}
    \begin{block}{ccc} & $j$ & $k$ \\ \end{block}
    \begin{block}{c[cc]}
       $j$ & 0 & 1\\ $k$ & -1 & 0 \\
    \end{block}
  \end{blockarray} & \begin{blockarray}{ccc}
    \begin{block}{ccc} & $j$ & $k$ \\ \end{block}
    \begin{block}{c[cc]}
       $j$ & 0 & 1\\ $k$ & -1 & 0 \\
    \end{block}
  \end{blockarray} & \begin{blockarray}{ccccc}
    & $j$ & $k$ & $j'$ & $k'$ \\
    \begin{block}{c[cc|cc]}
    $j$ & & $1$ & & \\
    $k$ & $\mi1$ & &  & \\ \cline{2-5}
    $j'$ & & & & $1$\\
    $k'$ &  & & $\mi1$ &\\
    \end{block}
    \end{blockarray}  & 
    \begin{blockarray}{ccccc}
    & $j$ & $k$ & $j'$ & $k'$ \\
    \begin{block}{c[cc|cc]}
    $j$ & & $1$ & & \\
    $k$ & $\mi1$ & & & \\ \cline{2-5}
    $j'$ & &  & & $1$\\
    $k'$ &  & & $\mi1$ &\\
    \end{block}
    \end{blockarray}  \\
    \hline
    $p_l$ & \begin{blockarray}{ccc}
    \begin{block}{ccc} & $j$ & $k$ \\ \end{block}
    \begin{block}{c[cc]}
       $j$ & 0 & 1\\ $k$ & 1 & 0 \\
    \end{block}
  \end{blockarray} & \begin{blockarray}{ccc}
    \begin{block}{ccc} & $j$ & $k$ \\ \end{block}
    \begin{block}{c[cc]}
       $j$ & 0 & $i$\\ $k$ & $i$ & 0 \\
    \end{block}
  \end{blockarray}& 
    \begin{blockarray}{ccccc}
    & $j$ & $k$ & $j'$ & $k'$ \\
    \begin{block}{c[cc|cc]}
    $j$ & & & & $1$\\
    $k$ & & & $1$ & \\ \cline{2-5}
    $j'$ & & $1$ & & \\
    $k'$ & $1$ & & &\\
    \end{block}
    \end{blockarray}  & \begin{blockarray}{ccccc}
    & $j$ & $k$ & $j'$ & $k'$ \\
    \begin{block}{c[cc|cc]}
    $j$ & & & & $1$\\
    $k$ & & & $1$ & \\ \cline{2-5}
    $j'$ & & $\mi1$ & & \\
    $k'$ & $\mi1$ & & &\\
    \end{block}
    \end{blockarray} \\
    \hline
    \end{tabular}
    \caption{Examples of eigenmatrices $x_l, \theta x_l$ and ping pong matrices $k_l, p_l$. $k_l = x_l + \theta x_l$ and $p_l = x_l - \theta x_l$ as defined in \eqref{eq:pingpongmatrix}. $k_l, p_l$ are normalized to have $\pm 1$ entries. A block structure on row/columns $j, k$ and $j' := p+j$ and $k':=p+k$ are filled up with $0$ and $\pm 1$.}
    \label{tab:kpexamples}
\end{table}

\par Matrixwise, the ping pong matrices $k_j\in\mathfrak{k}, p_j\in\mathfrak{p}$ of $\mathfrak{g}$ are brought back to a new set of ping pong matrices $k_j\in\mathfrak{k}_C, ip_j\in \mathfrak{p}_C$ in $\mathfrak{g}_C$. Let's denote them by $\tilde{k}_j:=k_j$ and $\tilde{p}_j:=ip_j$. The role of the subspace $\mathfrak{a}$ is now played by $i\mathfrak{a}$. replacing $\ad_H$ by $\ad_{iH}$. We deduce a set of similar relationships for $\tilde{k}_j, \tilde{p}_j$ under $\ad_{iH}$, 
\begin{equation*}
    \ad_{iH} (\tilde{k}_j) = \alpha_j\tilde{p}_j, \hspace{1cm} \ad_{iH}(\tilde{p}_j) = -\alpha_j\tilde{k}_j.
\end{equation*}
In matrix form,
\begin{equation}\label{eq:compactadH}
    \ad_{iH}\twoone{\tilde{k}_j}{\tilde{p}_j} = \twotwo{0}{\alpha_j}{-\alpha_j}{0}\twoone{\tilde{k}_j}{\tilde{p}_j},
\end{equation}
which leads to the compact version of \eqref{eq:expadHkp1} and \eqref{eq:expadHkp2},
\begin{equation}\label{eq:compactexpadH}
    \exp(\ad_{iH})\twoone{\tilde{k}_j}{\tilde{p}_j} = \twotworr{\cos\alpha_j}{\sin\alpha_j}{-\sin\alpha_j}{\cos\alpha_j}\twoone{\tilde{k}_j}{\tilde{p}_j}.
\end{equation}
At the group level,the symmetric spaces $G/K$ and $G_C/K$ are called the duals of each other, and they appear in the same row of standard symmetric space charts. An example of eigenmatrices $x_j, \theta x_j$ and ping pong matrices for some symmetric spaces and their duals are presented in Table \ref{tab:kpexamples}.

\subsection{Jacobian of the map $\Phi$} 
We provide a generalized algorithm for finding a Jacobian of the decomposition $\Phi(Q, \Lambda) = Q\Lambda Q^{-1}$ (as we defined in \eqref{eq:mapPhi}) where $\Lambda\in A:=\exp(\mathfrak{a}), Q\in K$.
The $\mathfrak{k}$ and $\mathfrak{p}$ from the previous section are the tangent spaces of $K$ and $P$, respectively. As mentioned, we follow Helgason's derivation \cite[Theorem 5.8 of Ch.\RN{1}]{Helgason1984} and start by directly translating his proof into simple linear algebra terms. In Table \ref{tab:comparison1}, we have Helgason's derivation (Left) compared in the same row with linear algebra (Right). Table \ref{tab:comparison1} is using the noncompact symmetric space $G/K$ but the compact case are identical with replacing $\sin\alpha_j$ by $\sinh\alpha_j$.  

\begin{table}[h]
\centering
{\tabulinesep=0.9mm
\small
\begin{tabu}{|l|l|}
\hline
\multicolumn{1}{|c|}{\makecell{Classical notation\\(\cite[p.187]{Helgason1984}, Pf of Theorem 5.8, Ch.$\RN{1}$)}} & \multicolumn{1}{|c|}{\makecell{Linear algebra notation\\ (Matrix factorizations)}}\\
\hline

\multicolumn{2}{|c|}{Definitions}\\
\hline

$\Phi :$ $K\times A\to G/K$ & $\tilde\Phi:$ $K\times A\to P$ \\

$\Phi :$ $(k, a) \mapsto kaK$ & $\tilde\Phi :$ $(Q, \Lambda)\mapsto Q\Lambda Q^{-1}$ ({\color{red}$\Lambda^\frac{1}{2} = a$, $Q = k$})\\

$d\tau(g_0):(G/K)_o\to (G/K)_{g_0\cdot o}$ &$d\tilde\tau(g_0):X\mapsto g_0X(\theta g_0)^{-1}$\\
$d\pi : \mathfrak{g}\to(G/K)_o$ &  \footnotesize{($\theta k = k, k\in K$, $\theta p = p^{-1}, p\in P$)}\\

At $k\in K$, fix a tangent vector $d\tau(k)T_i^\alpha$ & At $Q\in K$, fix a tangent vector $dQ$\\

At Id, basis element $T_i^\alpha\in \mathfrak{k}$ & At Id, basis element $Q^{-1}dQ = k_j\in\mathfrak{k}$\\
\hline

\multicolumn{2}{|c|}{Derivations}\\
\hline
$2d\Phi(d\tau(k)T_i^\alpha, 0)$\textsuperscript{$\star$}  & $d\tilde\Phi(dQ, 0) = d(Q\Lambda Q^{\mis 1})$  \footnotesize{({\color{red}With $d\Lambda=0$})}\\

$= d\pi(2kT_i^\alpha a)$ & $=dQ\Lambda Q^{\mis 1}+ Q\Lambda dQ^{\mis 1} $\\

$= d\tau(ka)d\pi(2\text{Ad}(a^{\mis 1})T_i^\alpha)$\textsuperscript{$\star\star$}& $=d\tilde\tau(Q\Lambda^\frac{1}{2})\big[\Lambda^{\mis\frac{1}{2}}(Q^{\mis 1}dQ\Lambda + \Lambda dQ^{\mis 1}Q)\Lambda^{\mis\frac{1}{2}}\big]$\textsuperscript{$\diamondsuit$}\\

$=d\tau(ka)d\pi(\text{Ad}(a^{\mis 1})T_i^\alpha - \text{Ad}(a)T_i^\alpha)$ & $=d\tilde\tau(Q\Lambda^\frac{1}{2}){\big[\Lambda^{\mis\frac{1}{2}}k_j\Lambda^\frac{1}{2} - \Lambda^\frac{1}{2}k_j\Lambda^{\mis\frac{1}{2}}\big]}$\\
 
\footnotesize{\big({\color{red} Let $H$ be such that $\exp(H) = a = \Lambda^{\frac{1}{2}}$}\big)} &\footnotesize{\big({\color{red}Note that $d\tilde\tau(Q\Lambda^\frac{1}{2})X = Q\Lambda^\frac{1}{2}X\Lambda^\frac{1}{2}Q^{\mis 1}$}\big)} \\

$=d\tau(ka)d\pi(e^{\mis \text{ad}H}T_{i}^\alpha - e^{\text{ad} H}T_i^\alpha)$  & 
$=d\tilde\tau(Q\Lambda^\frac{1}{2})[\exp(H^T\!\otimes\! I\! -\! I\!\otimes\! H)k_j$\\

& $\hspace{1.3cm}-\exp(I\!\otimes\! H \!- \!H^T\!\otimes \!I)k_j]$ (by \eqref{eq:expad})\\

$= d\tau(ka)d\pi(\mi \alpha(H)^{\mis 1}[H, T_i^\alpha]2\sinh\alpha(H))$ & $=d\tilde\tau(Q\Lambda^{\frac{1}{2}})\big[(-2\sinh\alpha_j) p_j\big]$ (by \eqref{eq:expadHpingpong})\\

\textsuperscript{$\star$}\scriptsize{Since $\Lambda^\frac{1}{2}=a$, we have $2d\Phi = d\tilde\Phi$} & \textsuperscript{$\diamondsuit$}\scriptsize{Both $dQ\Lambda Q^{\mis 1}$ and $Q\Lambda dQ^{\mis 1}$ are at $Q\Lambda Q^{\mis 1}$, should}\\

\textsuperscript{$\star\star$}\scriptsize{This is $(d\tau(ka)\circ d\pi)(\text{Ad}(a^{\mis 1})T_i^\alpha$}) & \scriptsize{be brought back to identity (inside bracket).}\\
\hline
\end{tabu}}
\caption{\label{tab:comparison1}Line-by-line translation of the classical proof to linear algebra proof}
\vspace{-0.5cm}
\end{table}
\par From the last line of Table \ref{tab:comparison1} we can finish the story with two different directions, depending on the choice of the volume measure. First, if we use a \textbf{$\boldsymbol{G}$-invariant measure} (the ``canonical measure") of $P$, the measure is invariant under the map $d\tau$ or $d\tilde\tau$ (by definition of the invariant measure). Thus we can disregard $d\tilde\tau(Q\Lambda^\frac{1}{2})$(or $d\tau(ka)$) so that the Jacobian of $d\tilde\Phi$ (or $d\Phi$) only depends on the differential map $k_j\mapsto(\sinh\alpha_j) p_j$. Since $\cup_j \{k_j\}$ and $\cup_j \{p_j\}$ are both orthonormal bases, we obtain the Jacobian \eqref{eq:noncompactk1ak2jac}
\begin{equation*}
    \prod_{\alpha\in\Sigma^+}\sinh\alpha(H).
\end{equation*}
Note that eigenvalues $\pm\alpha_j$ belong to $x_j$ and $\theta x_j$, has the same corresponding $k_j$. (See \eqref{eq:pingpongmatrix} and above.) Thus we only take the positive roots $\Sigma^+$ above. 

\par The second choice of measure is the \textbf{Euclidean measure}, which is a wedge product of independent entrywise differentials. In this case the procedure is identical up to the factor $\sinh\alpha_j$, but the map $d\tilde\tau(Q\Lambda^\frac{1}{2})$ (equivalently $d\tau(ka)$) cannot be ignored. One needs to carefully compute the differential map $d\tilde\tau(Q\Lambda^\frac{1}{2})p_j = Q\Lambda^\frac{1}{2}p_j\Lambda^\frac{1}{2}Q^{-1}$ under the Euclidean measure. We can further use the fact that conjugation by the  matrix $Q$ always preserves the Euclidean measure, since the subgroup $K$ is always a set of matrices with an Orthogonal/Unitary type of property. Thus, one needs to compute the map $p_j \mapsto\Lambda^\frac{1}{2}p_j\Lambda^\frac{1}{2}$ and multiply its Jacobian by $\prod_{\alpha\in\Sigma^+}\sinh\alpha(H)$. 

\begin{remark}
For the compact Lie group $G$, we have $\sinh\alpha_j$ replaced by $\sin\alpha_j$ everywhere. Moreover, the last Jacobian computation step $p_j\mapsto\Lambda^\frac{1}{2} p_j \Lambda^\frac{1}{2}$ can be omitted for the compact cases, since $\Lambda^\frac{1}{2}$ is an orthogonal/unitary matrix for the compact cases. The map $d\tilde\tau(\Lambda^\frac{1}{2})$ preserves the Euclidean measure as $d\tilde\tau(Q)$. 
\end{remark}

\subsection{Extension to the generalized Cartan decomposition}
In the previous paragraphs, we studied the Jacobian of the usual Cartan decomposition. We now proceed to consider the generalized Cartan decomposition (Theorems \ref{thm:noncompactk1ak2decomp} and \ref{thm:compactk1ak2decomp}), its Jacobian \eqref{eq:noncompactk1ak2jac}, \eqref{eq:compactk1ak2jac} and the extension of Table \ref{tab:comparison1}. The derivations are analogous, analyzing subspaces of $\mathfrak{g}$ but one should now proceed with four tangent subspaces, $\mathfrak{k}_\tau\cap\mathfrak{k}_\sigma$, $\mathfrak{k}_\tau\cap\mathfrak{p}_\sigma$, $\mathfrak{p}_\tau\cap\mathfrak{k}_\sigma$, $\mathfrak{p}_\tau\cap\mathfrak{p}_\sigma$. Earlier work on these Jacobian related derivations may be found in \cite{flensted1980discrete,hoogenboom1983generalized}. The maximal subspace $\mathfrak{a}$ is now defined inside $\mathfrak{p}_\tau\cap\mathfrak{p}_\sigma$. We start with the same strategy: the tangent space $\mathfrak{g}$ is decomposed into the eigenspaces of the linear operator $\ad_H$ with $H\in \mathfrak{a}$. The eigenvalues $\pm\alpha_j$ still come in pairs but we have two eigenmatrices $x_j, \tau\sigma x_j$ for eigenvalue $\alpha_j$, and two eigenmatrices $\tau x_j, \sigma x_j$ for eigenvalue $-\alpha_j$. We define four vectors $v_1, v_2, w_1, w_2$ with the same roles as $k_j$ and $p_j$ played before, 
\begin{gather*}
    v_1:=x_j + \tau x_j + \sigma x_j + \tau\sigma x_j\in\mathfrak{k}_\tau\cap\mathfrak{k}_\sigma,\hspace{0.5cm}v_2:=x_j - \tau x_j - \sigma x_j + \tau\sigma x_j\in\mathfrak{p}_\tau\cap\mathfrak{p}_\sigma\\
    w_1:=x_j - \tau x_j + \sigma x_j - \tau\sigma x_j\in\mathfrak{p}_\tau\cap\mathfrak{k}_\sigma,\hspace{0.5cm}w_2:=x_j + \tau x_j - \sigma x_j - \tau\sigma x_j\in\mathfrak{k}_\tau\cap\mathfrak{p}_\sigma
\end{gather*}
and these have similar ping pong relationships by $\ad_H$ like $k_j$ and $p_j$, 
\begin{gather*}
    \ad_H(v_1) = \alpha_j v_2\hspace{1cm}\ad_H(v_2) = \alpha_j v_1\\
    \ad_H(w_1) = \alpha_j w_2\hspace{1cm}\ad_H(w_2) = \alpha_j w_1.
\end{gather*}
We can similarly extend \eqref{eq:expadHpingpong} and other relationships, and proceed as in Table \ref{tab:comparison1} to obtain \eqref{eq:noncompactk1ak2jac} and \eqref{eq:compactk1ak2jac}.

\section{Random matrix ensembles: compact and noncompact}
\subsection{Compact symmetric spaces}
In compact cases the random matrices could be simply determined from the Haar measure of the compact Lie group $G$ \cite{duenezthesis,duenez2004random}, since the compactness of $G$ turns the Haar measure into a probability measure. In the following sections we discuss random matrix ensembles based on 10 types of Riemannian symmetric space classification by Cartan. For the triple $(G, K_\sigma, K_\tau)$, we start with the cases where $G/K_\sigma$ and $G/K_\tau$ are of the same types in Sections \ref{sec:circular} and \ref{sec:jacobi}. Then, in Section \ref{sec:mixed} we will discuss the ``mixed types" where $G/K_\sigma$ and $G/K_\tau$ are different types under Cartan's classification. 

\subsection{Noncompact symmetric spaces}\label{sec:noncompactRMT}
Sections \ref{sec:hermite} and \ref{sec:Laguerre} discuss classical random matrix ensembles associated to noncompact symmetric spaces. Hermite and Laguerre eigenvalue joint densities arise as result of \eqref{eq:noncompactk1ak2jac}, using Theorem \ref{thm:symspacejacobians} on noncompact symmetric spaces. As opposed to compact Lie groups and symmetric spaces where the Haar measure or $G$-invariant measure can be normalized by a constant to a probability measure, invariant measures on noncompact manifolds cannot be normalized to one by constants. A normalizing factor $\mathcal{S}$ should be introduced to complete the construction of a probability measure. Therefore, random matrices on a noncompact manifold face an innate problem if we proceed analogous to Sections \ref{sec:circular} and \ref{sec:jacobi}:
\begin{itemize}
    \item The choice of the probability measure on noncompact $G/K$ is not unique. 
\end{itemize}
In \cite{duenez2004random}, Due{\~n}ez also addresses this problem along the noncompact duals.

\par As we push the measure forward to the subgroup $A$, the resulting measure should be a symmetric function of independent generators of $A$. Hence the probability measure $\mathcal{I}(g)$ of the random matrix ensemble is the Haar or $G$-invariant measure on $G$ or $G/K$, multiplied by some symmetric function $\mathcal{S}$ on $A$,  
\begin{equation*}
    \mathcal{I}(g) = \mathcal{S}(a)\mu(g),
\end{equation*}
where $g = k_1ak_2$ or $g = kak^{-1}$ and $\mu(g)$ is an invariant measure. Using \eqref{eq:noncompactk1ak2jac}, the measure on $A$ is induced, 
\begin{equation*}
    \mathcal{I}(g) = dk\cdot\mathcal{S}(a)\bigg(\prod_{\alpha\in\Sigma^+}\sinh\alpha(H)\bigg)dH_1\cdots dH_{\text{dim(A)}},
\end{equation*}
which means even though the measure $\mathcal{I}$ changes, the measure on $A$ still differs only by a normalization function. The traditional choice of $\mathcal{S}$ has been made such that $\mathcal{I}(g)$ can be constructed from independent Gaussian distributions endowed on matrix entries. In fact, one could also endow a Gaussian distribution on the Riemannian manifold (symmetric space) itself \cite{heuveline2021gaussian}. 

\par An alternative approach which appears in \cite{altland1997nonstandard} is to put a probability measure on the tangent space of the symmetric space, $\mathfrak{p}$. In particular, independent Gaussian distribution endowed on the elements of $\mathfrak{p}$ give rise to Hermite and Laguerre ensembles by Theorem \ref{thm:liealgebrajac}. We will follow this alternative approach. 

\subsection{Non-probability measure of noncompact groups}\label{sec:noncompactHSVD}
As discussed in Section \ref{sec:noncompactRMT}, the Haar measure of a noncompact group $G$ or a noncompact symmetric space $G/K$ is not a probability measure. However, we can force an analogue of a random matrix theory. Imagine for example a noncompact $\kak$ decomposition $G = K_\sigma A K_\tau$ with $(G, K_\sigma, K_\tau) = (\gl{n, \mathbb{R}}, \ortho{n}, \ortho{p, q})$. This is called the Hyperbolic SVD \cite{onn1989hyperbolic} where any real invertible matrix $M$ is factored into the product of an orthogonal matrix $O$, a positive diagonal matrix $\Lambda$ and an indefinite orthogonal matrix $V$. From the Haar measure and \eqref{eq:noncompactk1ak2jac} of $\gl{n, \mathbb{R}}$ one obtains the Jacobian,
\begin{equation*}
    \prod_{\substack{1\leq j < k \leq p \\ p<j<k\leq n}} |\lambda_j -\lambda_k| \prod_{\substack{1\leq j \leq p\\p<j\leq n}} |\lambda_j + \lambda_k| \prod_{j=1}^n |\lambda_j|^{-\frac{2n+1}{2}} d\lambda_1 \dots d\lambda_n,
\end{equation*}
where $\lambda_j$ is the squared diagonal entries of $\Lambda$ for all $j$'s. 

\par One can impose a Gaussian-like density function (although not a probability density) on the group $\gl{n, \mathbb{R}}$, such as $\exp(-\tr(gI_{p,q}g^T)/2)\prod dg_{jk}$, where $I_{p, q} = \text{diag}(I_p, -I_q)$. In terms of independent entries of $g$ this is
\begin{equation}\label{eq:hsvdrmt}
    \prod_{\text{first $p$ columns}} e^{-g_{jk}^2/2} \prod_{\text{last $q$ columns}} e^{g_{jk}^2/2} \prod dg_{jk}. 
\end{equation}
Since the Haar measure of $\gl{n, \mathbb{R}}$ is $|\det(g)|^{-n}\prod dg_{jk}$, \eqref{eq:hsvdrmt} becomes (after integrating out $\ortho{n}$ and $\ortho{p, q}$), 
\begin{equation*}
    \prod_{j<k}|\lambda_j - \lambda_k| \prod_{j=1}^n|\lambda|^{-\frac{n+1}{2}}e^{-\sum\lambda_j/2} d\lambda_1\dots d\lambda_n,
\end{equation*}
where $\lambda_1, \dots, \lambda_p\geq 0$ are the first $p$ squared diagonal values of $\Lambda$ and $\lambda_{p+1},\dots,\lambda_n\leq 0$ are the last $q$ squared diagonal values of $\Lambda$, multiplied by $-1$. Extending this approach to find a proper random matrix probability measure on noncompact Lie groups and symmetric spaces with joint probability densities on the subgroup $A$, is still an open problem.

\section{Compact A$\RN{1}$, A, A$\RN{2}$: Circular ensembles}\label{sec:circular}

The joint probability density of the circular ensemble is ($\beta=1,2,4$),
\begin{equation*}
    E_n^{(\beta)}(\theta) \propto \prod_{j<k}|e^{i\theta_j} - e^{i\theta_k}|^\beta.
\end{equation*}
Circular ensembles $\beta = 1, 2, 4$ (COE, CUE, CSE) arise as the eigenvalues of special unitary matrices. As we discuss in the introduction, circular ensembles are completely classified by (compact) symmetric spaces of the types A$\RN{1}$, A and A$\RN{2}$, respectively \cite{dyson1970correlations,duenez2004random}. The $\kak$ decomposition associated to each symmetric space recovers the KAK decomposition. The restricted root system (and dimensions) of A$\RN{1}$, A, A$\RN{2}$ are given as the following: ($1\le j< k \le n$)
\begin{equation}\label{tab:circularroots}
    \renewcommand\arraystretch{1.2}
    \begin{array}{r|c|}
    \cline{2-2}
    \alpha(H) & \pm (h_j-h_k)\\
    \cline{2-2}
    m_{\alpha} & \beta\\
    \cline{2-2}
    \end{array}
\end{equation}
Since we have compact symmetric spaces, we use \eqref{eq:compactk1ak2jac} from either Theorem \ref{thm:compactk1ak2decomp} or \ref{thm:symspacejacobians} with these root systems.

\subsection{Compact A$\RN{1}$, $\beta=1$ COE}\label{sec:COE}
The compact symmetric space A$\RN{1}$ is $G/K = \un{n}/\ortho{n}$. The involution on $\un{n}$ has no free parameter and the $\kak$ decomposition is equivalent to the KAK decomposition of $\un{n}/\ortho{n}$. (In other words, we only have Cartan's coordinate system.) The maximal abelian torus $A$ is,
\begin{equation*}
A = \{\text{Diagonal matrices with entries $e^{ih_j}$, where $h_j\in\mathbb{R}$}\}.
\end{equation*}
From the KAK decomposition, we obtain $U = O_1DO_2$, a factorization of a unitary matrix $U$ into the product of two orthogonal matrices $O_1, O_2\in\ortho{n}$ and a unit complex diagonal matrix $D\in A$. This decomposition first appears in \cite{fuehr2018note} and we will call this the \textit{ODO decomposition}. The corresponding Jacobian (up to constant) from \eqref{eq:compactk1ak2jac} using \eqref{tab:circularroots}, $\beta = 1$ is (with the change of variables $\theta_j = 2h_j$),
\begin{equation*}
\bigg(\prod_{j<k}\sin(h_j - h_k)\bigg) dh_1\cdots dh_n\,\,\, \propto\,\,\, \prod_{j<k}\vert e^{i\theta_j} - e^{i\theta_k}\vert d\theta_1\cdots d\theta_n.
\end{equation*}
This is the joint density of the COE. In other words, doubled angles in the diagonal of $D$ from the ODO decomposition of a Haar distributed unitary matrix is the COE distribution. Moreover if we identify $G/K$ as the set of unitary symmetric matrices $P$, the map \eqref{eq:mapPhi} is the factorization $S = O\Lambda O^T$, the eigendecomposition of a unitary symmetric matrix $S$ with real eigenvectors $O$. In terms of Remark \ref{rem:representations}, $U = O_1DO_2$ becomes $S = UU^T = O_1D^2O_1^T$ where $\Lambda = D^2$. To obtain the COE, we can utilize both factorizations:
\begin{itemize}
  \item Two times the angles of the unit diagonal values of $D$ from the ODO decomposition of $U\in\text{Haar}(\un{n})$. 
  \item The angles of the (unit) eigenvalues of a unitary symmetric matrix obtained from $UU^T$, $U\in\text{Haar}(\un{n})$.
\end{itemize}
\begin{remark}
The second algorithm above would be obvious since the days of Dyson \cite{dyson1962a,dyson1962threefold} while we are not aware of the first algorithm appearing in the literature. 
\end{remark}

\subsection{Compact A, $\beta = 2$ CUE}
The symmetric space of the compact type A is $G/K = \un{n}\times\un{n}/\un{n}$. The restricted root system returns to the usual root system $A_n$ of the classical semisimple Lie algebra. A maximal torus of $\un{n}$ is a Cartan subalgebra of $\un{n}$. Weyl's integration formula agrees with \eqref{eq:compactk1ak2jac} obtaining the CUE, which is the eigenvalues of a Haar distributed unitary matrix. The derivation of the CUE can be found in many random matrix textbooks \cite{blower2009random,forrester2010log,mehta2004random}. 

\subsection{Compact A$\RN{2}$, $\beta = 4$ CSE}\label{sec:CSE}
The involution $X\mapsto -J_n^TX^{T}J_n$ where $J_n:=\begin{bsmallmatrix} 0 & I_n\\-I_n& 0\end{bsmallmatrix}$ on the tangent space of $\un{2n}$ results the symmetric space $\un{2n}/\symp{n}$ where $\symp{n}=\symp{2n, \mathbb{C}}\cap\un{2n}$. A choice of maximal abelian torus $A$ is
\begin{equation*}
    A = \{\text{diag}(\Tilde{D}, \Tilde{D}):\tilde{D} = \text{diag}(e^{ih_1}, \dots, e^{ih_n}),\,\, h_j\in\mathbb{R}\}.
\end{equation*}
Again from the KAK decomposition, we obtain $U = Q_1DQ_2$, a factorization of a $2n\times 2n$ unitary matrix $U$ into the product of two unitary symplectic matrices $Q_1, Q_2\in\symp{n}$ and a unit complex diagonal matrix $D\in A$. We call this the \textit{QDQ decomposition}. The corresponding Jacobian from \eqref{eq:compactk1ak2jac} using \eqref{tab:circularroots}, $\beta = 4$ is,
\begin{equation*}
\bigg(\prod_{j<k}\sin^4(h_j - h_k)\bigg) dh_1\cdots dh_n\,\,\,\propto\,\,\,\prod_{j<k}\vert e^{i\theta_j} - e^{i\theta_k}\vert^4 d\theta_1\cdots d\theta_n,
\end{equation*}
with the change of variables $\theta_j = 2h_j$. This is the CSE distribution. Similarly as in Section \ref{sec:COE}, the eigendecomposition of unitary skew-Hamiltonian matrix obtained by $UJ_nU^TJ_n^T$, $U\in\text{Haar}(2n)$ is equivalent to the map \eqref{eq:mapPhi}. Two numerical algorithms for sampling the CSE are the following:
\begin{itemize}
    \item Two times the angles of the first $n$ unit diagonal values of $D$ from the QDQ decomposition of $U\in\text{Haar}(\un{2n})$.
    \item The angles of the first $n$ (unit) eigenvalues of a unitary skew-Hamiltonian matrix obtained by $UJ_nU^TJ_n^T$ with $U\in\text{Haar}(\un{2n})$.
\end{itemize}

\section{Compact BD$\RN{1}$, A$\RN{3}$, C$\RN{2}$: Jacobi ensembles}\label{sec:jacobi}
The joint probability density of the Jacobi ensemble is ($\beta = 1,2,4$), 
\begin{equation*}
    J_{\alpha_1, \alpha_2}^{(\beta), m}(x) \propto \prod_{j<k}|x_j - x_k|^\beta\prod_{j=1}^m x_j^{\alpha_1}(1-x_j)^{\alpha_2}.
\end{equation*}
In \cite{duenezthesis,duenez2004random}, Jacobi ensembles $\beta = 1, 2, 4$ arise from the KAK decompositions of seven compact symmetric spaces, BD$\RN{1}$, A$\RN{3}$, C$\RN{2}$, D$\RN{3}$, BD, C, C$\RN{1}$. Especially, types BD$\RN{1}$, A$\RN{3}$, C$\RN{2}$ give multiple Jacobi densities as follows: (for integeres $p\geq q$)
\begin{equation*}
    \prod_{j<k}|x_j - x_k|^\beta \prod_{j=1}^q x_j^{\frac{\beta}{2}-1}(1-x_j)^{\frac{\beta(p-q+1)}{2}-1},
\end{equation*}
and the powers of $x_j$'s are fixed to $\frac{\beta}{2}-1$. The remaining four cases add four more parameter points, which could be found in \cite{duenez2004random,duenezthesis}. In this paper we omit these four cases as these do not have any further results, as they only have Cartan's coordinates (no free parameter for the Cartan involution). 

\par The $\kak$ decomposition $G=K_\tau A K_\sigma$ of the compact types BD$\RN{1}$-$\RN{1}$, A$\RN{3}$-$\RN{3}$, C$\RN{2}$-$\RN{2}$ are exactly the \textit{CS decomposition} (CSD) \cite{davis1969some,davis1970rotation} of orthogonal, unitary, unitary symplectic matrices, respectively. The decomposition $\Phi$ of the symmetric space (Theorem \ref{thm:symspacejacobians}) is the GSVD coordinate systems we discussed in Sections \ref{sec:introGSVD} and \ref{sec:onevsmany}. Assume $r \ge p \ge q\ge s$ and $n = p+q = r+s$ throughout this section. We note that with the KAK decomposition, only the cases $p=r, q=s$ are obtained for the CSD. The root system associated to the $\kak$ decomposition is the following ($1\le j< k\le s$). 
\begin{equation}\label{tab:jacobiroots}
    \renewcommand\arraystretch{1.2}
    \begin{array}{r|c|c|c|}
    \cline{2-4}
    \alpha(H) & \pm (\theta_j\pm \theta_k) & \pm \theta_j & \pm 2\theta_j\\
    \cline{2-4}
    m_{\alpha}^+ & \beta  & \beta(p-s) & \beta-1\\
    \cline{2-4}
    m_{\alpha}^- & 0  & \beta(q-s) & 0\\
    \cline{2-4}
    \end{array}
\end{equation}
For all three $\beta$ we have the identical maximal abelian subgroup $A$,
\begin{equation*}
    A = \{\text{$n\times n$ matrices with the block structure} \begin{bsmallmatrix}
\,\,C & & S \\
 & I_{p-q} & \\
-S & & C
\end{bsmallmatrix}\}
\end{equation*}
where $C, S\in\mathbb{R}^{s\times s}$ are diagonal matrices with cosine, sine values of $\theta_1, \dots, \theta_s$ on diagonal entries, respectively

\subsection{Compact BD$\RN{1}$-$\RN{1}$, $\beta = 1$ Jacobi}
With the involution $X\mapsto I_{p,q}XI_{p,q}$ on the tangent space of $\ortho{n}$ we obtain the symmetric space BD$\RN{1}$, $G/K = \ortho{n}/(\ortho{p}\times\ortho{q})$, where $I_{p,q}:= \text{diag}(I_p, -I_q)$. With two symmetric pairs $(\ortho{n}, \ortho{p}\times\ortho{q})$ and $(\ortho{n}, \ortho{r}\times\ortho{s})$, we obtain the $\kak$ decomposition\footnote{Equivalently, one can imagine the GSVD of \eqref{eq:introgsvdcoordinates}.} BD$\RN{1}$-$\RN{1}$: 
\begin{equation*}
\begin{bmatrix}\text{$n$-by-$n$}\\\text{Orthogonal}\end{bmatrix}
= \begin{bmatrix}
O_p & \\ & O_q
\end{bmatrix}
\begin{bmatrix}
C & & S \\
 & I_{n-2s} & \\
-S & & C
\end{bmatrix}
\begin{bmatrix}
O_r & \\& O_s \\
\end{bmatrix}.
\end{equation*}
This is the real CSD. From \eqref{eq:compactk1ak2jac} using \eqref{tab:jacobiroots} $\beta=1$, we obtain the Jacobian   
\begin{equation*}
    d\mu(H) \propto \prod_{j<k}\big(\sin(\theta_j-\theta_k)\sin(\theta_j+\theta_k)\big)\prod_j\big((\sin \theta_j)^{(p-s)}(\cos \theta_j)^{(q-s)}\big)d\theta_1\dots d\theta_s.
\end{equation*}
Using trigonometric identities with change of variables $x_j = \cos^2\theta_j = \frac{1+\cos(2\theta_j)}{2}$,
\begin{equation*}
d\mu(H) \propto \prod_{j<k} \vert x_j - x_k\vert \prod_{j=1}^s x_j^{\frac{1}{2}(q-s+1)-1}(1-x_j)^{\frac{1}{2}(p-s+1)-1} dx_1\dots dx_s.
\end{equation*}
which is the joint density of the $\beta=1$ Jacobi ensemble $J_{\alpha_1, \alpha_2}^{(1), s}$ if we let $\alpha_1 = \frac{1}{2}(q-s+1)-1$, $\alpha_2 = \frac{1}{2}(p-s+1)-1$. This result agrees with \cite[Theorem 1.5]{edelman2008beta}, where the squared CSD cosine values of a Haar distributed orthogonal matrix are distributed as $\beta=1$ Jacobi ensemble. Moreover, recall the fact that the QL decomposition $G = QL$ (a lower triangular analogue of the QR decomposition) of an $n\times n$ independent Gaussian matrix $G$ obtains a Haar distributed orthogonal matrix $Q$. Since the GSVD \cite{paige1981gsvd,van1976gsvd} is equivalent to the combination of the QL decomposition and the CSD, one can take the GSVD of a real independent Gaussian matrix to obtain the same $\beta=1$ Jacobi ensemble. Two associated numerical algorithms are the following. ($a = q-s, b = p-s$)
\begin{itemize}
    \item The squared CSD cosine values of a Haar distributed $m\times m$ orthogonal matrix $(m = 2s+a+b)$ with row/column partitions $(s+a, s+b)$ and $(s, s+a+b)$. 
    \item The squared cosine values, where the tangent values are the generalized singular values of real $(s+a)\times s$ and $(s+b)\times s$ Gaussian matrices. 
\end{itemize}

\subsection{Compact A$\RN{3}$-$\RN{3}$, $\beta = 2$ Jacobi}\label{sec:jacobibeta2}
Two symmetric pairs of compact $A\RN{3}$ type are $(\un{n}, \un{p}\times\un{q})$ and $(\un{n}, \un{r}\times\un{s})$. The $\kak$ decomposition of the group $G$ is the CSD of unitary matrices, and the decomposition of $G/K_\sigma = \un{n}/(\un{r}\times \un{s})$ is the complex GSVD described in Section \ref{sec:introGSVD} and equation \eqref{eq:introgsvdcoordinates}. Using \eqref{eq:compactk1ak2jac} with the root system \eqref{tab:jacobiroots}, $\beta=2$ and change of variables $x_j = \cos^2\theta_j$ as above, we obtain the Jacobian,
\begin{align*}
    \prod_{j<k}\big(\sin(\theta_j-\theta_k)&\sin(\theta_j+\theta_k)\big)^2\prod_j\big((\sin \theta_j)^{2(p-s)}(\cos \theta_j)^{2(q-s)} \sin(2\theta_j)\big)d\theta_1\dots d\theta_s\\
    \propto \,\,\,&\prod_{j<k} \vert x_j - x_k\vert^2 \prod_j x_j^{q-s}(1-x_j)^{p-s}  dx_1\dots dx_s,
\end{align*}
which is the $\beta=2$ Jacobi density $J_{\alpha_1, \alpha_2}^{(2), s}$ with $\alpha_1 = q-s, \alpha_2 = p-s$. Numerically the following could be utilized to obtain $\beta=2$ Jacobi densities.  ($a = q-s, b = p-s$)  
\begin{itemize}
    \item The squared CSD cosine values of a Haar distributed $m\times m$ unitary matrix $(m=2s+a+b)$ with row/column partitions $(s+a, s+b)$ and $(s, s+a+b)$.
    \item The squared cosine values, where the tangent values are the generalized singular values of complex $(s+a)\times s$ and $(s+b)\times s$ Gaussian matrices. 
\end{itemize}

\subsection{Compact C$\RN{2}$-$\RN{2}$, $\beta = 4$ Jacobi}
Jacobi densities with $\beta=4$ are similarly obtained from two symmetric spaces $\symp{n}/(\symp{p}\times\symp{q})$ and $\symp{n}/(\symp{r}\times\symp{s})$, where both are compact type C$\RN{2}$. We identify $\symp{n}$ as the quaternionic unitary group, $\un{n, \mathbb{H}}:=\{g\in\gl{n, \mathbb{H}}|g^Dg = I_n\}$. The $\kak$ decomposition is the CSD of a quaternionic unitary matrix. Using \eqref{eq:compactk1ak2jac} with the root system \eqref{tab:jacobiroots} $\beta=4$, we obtain the following Jacobian with the change of variables $x_j = \cos^2\theta_j$,
\begin{align*}
    \prod_{j<k}\big(\sin(\theta_j-\theta_k)&\sin(\theta_j+\theta_k)\big)^4\prod_j\big((\sin \theta_j)^{4(p-s)}(\cos \theta_j)^{4(q-s)} \sin^3(2\theta_j)\big)d\theta_1\dots d\theta_s\\
    \propto &\,\,\, \prod_{j<k} \vert x_j - x_k\vert^4 \prod_j x_j^{2(q-s)+1}(1-x_j)^{2(p-s)+1} dx_1\dots dx_s,
\end{align*}
which is the $\beta=4$ Jacobi density $J_{\alpha_1, \alpha_2}^{(4), s}$ with $\alpha_1=2(q-s)+1, \alpha_2 = 2(p-s)+1$. The associated numerical algorithm is the following. ($a = q-s, b = p-s$)
\begin{itemize}
    \item The squared cosine CS values of a Haar distributed $m\times m$ quaternionic unitary matrix $(m=2s+a+b)$ with row/column partitions $(s+a, s+b)$ and $(s, s+a+b)$.
\end{itemize}
\begin{remark}
Again, one can use the GSVD on quaternionic Gaussian matrices to obtain the classical $\beta = 4$ Jacobi ensemble.
\end{remark}

\section{Compact mixed types: More circular and Jacobi}\label{sec:mixed}
In this section we show even more cases such that a single symmetric space leading to multiple random matrix theories. We introduce $\kak$ decompositions with two compact symmetric spaces, each from different Cartan types. The classification of such $\kak$ decompositions is studied in \cite{matsuki2002classification}, with the computation of corresponding root systems. As always the names of these decompositions are combinations of two Cartan types, i.e., A$\RN{1}$-$\RN{2}$ represents $(G, K_\sigma, K_\tau) = (\un{2n}, \ortho{2n}, \symp{2n})$. 

\begin{figure}[h]
    \centering
    {\large Possible parameters $(\alpha_1, \alpha_2)$ of the $\beta=2$ Jacobi ensemble \\ \large\vspace{-0.6cm} $$J_{\alpha_1, \alpha_2}(x) \sim \prod_{j<k}|x_j - x_k|^2 \prod_{j=1}^m x_j^{\alpha_1}(1-x_j)^{\alpha_2}$$\vspace{-0.2cm}\\}
    \includegraphics[width=4.2in]{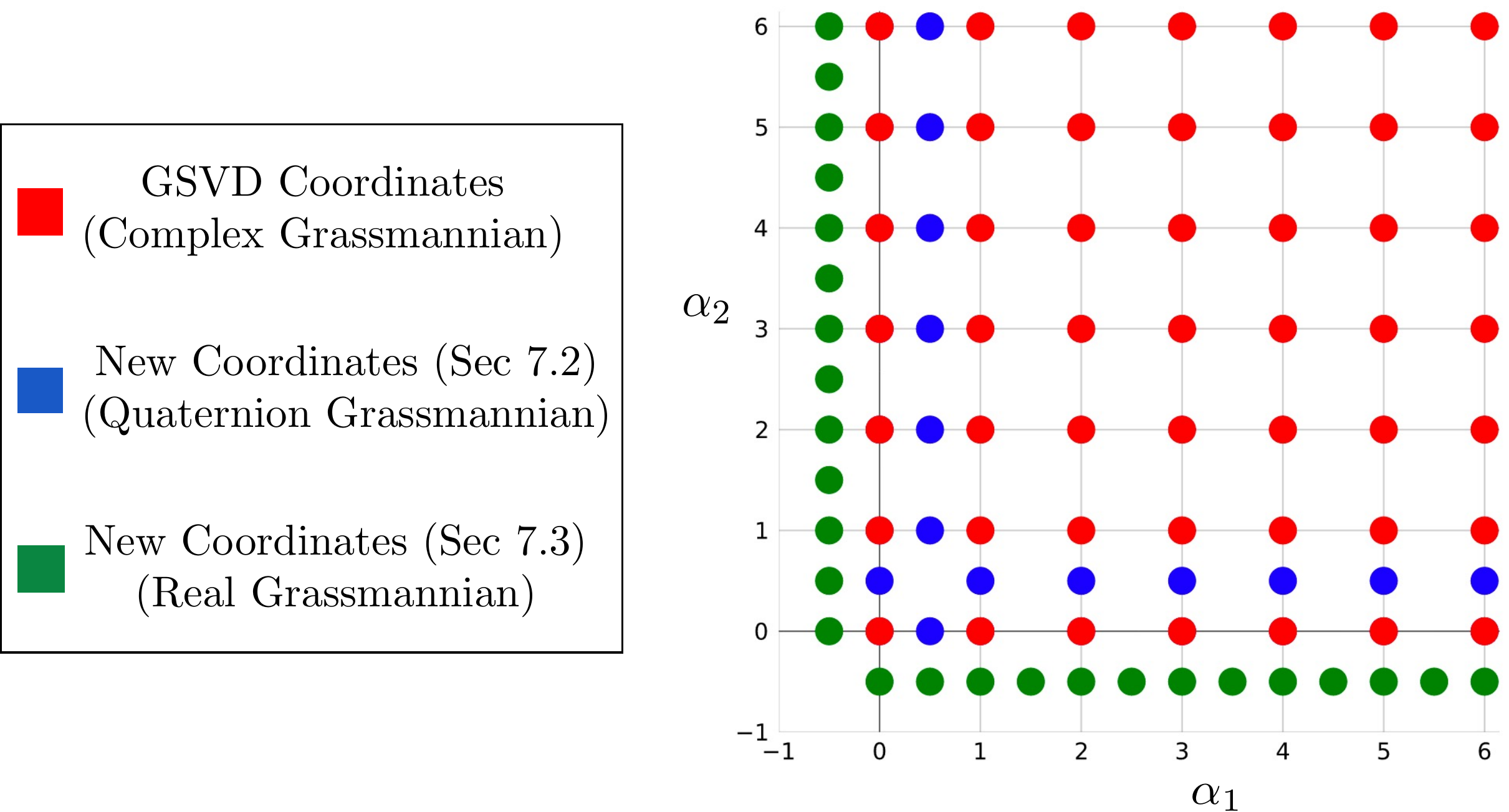}
    \caption{The parameter space $(\alpha_1, \alpha_2)\in(-1, \infty)^2$ of the $\beta=2$ Jacobi ensemble covered by symmetric spaces. The GSVD coordinate systems on the complex Grassmannian manifold (A$\RN{3}$-$\RN{3}$) discussed in Section \ref{sec:jacobi} covers red dots. A new coordinate system on the quaternionic (resp. real) Grassmannian manifold discussed in Section \ref{sec:newjacobi1} (resp. \ref{sec:newjacobi2}) of type C$\RN{1}$-$\RN{2}$ (resp. D$\RN{1}$-$\RN{3}$) represent blue (resp. green) dots.}
    \label{fig:newjacobi2}
\end{figure}

\subsection{Compact A$\RN{1}$-$\RN{2}$}
The two compact symmetric spaces are types A$\RN{1}$ and A$\RN{2}$, $\un{2n}/\ortho{2n}$ and $\un{2n}/\text{Usp}(2n)$. A maximal abelian subalgebra $\mathfrak{a}\subset\mathfrak{p}_\sigma \cap\mathfrak{p}_\tau$ is the set of all matrices $\text{diag}(i\theta_1, \dots, i\theta_n, i\theta_1, \dots, i\theta_n)$ for $(\theta_1, \dots, \theta_n)\in\mathbb{R}^n$. The subgroup $A$ is the following:
\begin{equation*}
    A = \{\text{diag}(\Tilde{D}, \Tilde{D}):\tilde{D} = \text{diag}(e^{i\theta_1}, \dots, e^{i\theta_n})\}.
\end{equation*}
The root system is given as
\begin{equation}\label{eq:a12roots}
    \renewcommand\arraystretch{1.3}
    \begin{array}{r|c|}
    \cline{2-2}
    \alpha(H) & \pm(\theta_j-\theta_k) \\
    \cline{2-2}
    m_{\alpha}^+ & 2\\
    \cline{2-2}
    m_{\alpha}^- & 2\\
    \cline{2-2}
    \end{array}
\end{equation}
Using \eqref{eq:compactk1ak2jac}, we obtain the Jacobian ($\xi_j = 4\theta_j$)
\begin{equation*}
    |e^{i\xi_j} - e^{i\xi_k}|^2d\xi_1 \cdots d\xi_n,
\end{equation*}
which is the joint probability density of the CUE. Hence, we obtain another sampling method for the CUE.

\subsection{Compact A$\RN{1}$-$\RN{3}$, C$\RN{1}$-$\RN{2}$}\label{sec:newjacobi1}
The two symmetric spaces in each case is the following: 
\begin{align*}
    G/K_\tau, G/K_\sigma &= \un{n}/\ortho{n}, \un{n}/(\un{p}\times \un{q}) \\
    G/K_\tau, G/K_\sigma &= \un{n, \mathbb{H}}/\un{n}, \un{n, \mathbb{H}}/(\un{p, \mathbb{H}}\times \un{q, \mathbb{H}})
\end{align*}
The subgroup $A$ is computed as follows.
\begin{equation*}
    A = \{\text{$n\times n$ matrices with the block structure} \begin{bsmallmatrix}
\,\,C & & \eta S \\
 & I_{p-q} & \\
\eta S & & C
\end{bsmallmatrix}\},
\end{equation*}
where $C, S$ are $q\times q$ diagonal matrices with cosine and sine values of $q$ angles $\theta_1, \dots, \theta_q$ on their diagonals. The imaginary unit $\eta$ is $i$ for A$\RN{1}$-$\RN{3}$ ($\beta=1$) and $\eta = j, k$ for C$\RN{1}$-$\RN{2}$ ($\beta=2$).\footnote{In fact, if we select the subgroup $K$ of $\un{n, \mathbb{H}}/\un{n}$ to be the unitary group with the imaginary unit $j$, we can also obtain $\eta = i$.} The root system is the following. ($\beta=1, 2$)
\begin{equation}\label{eq:newjacobiroots1}
    \renewcommand\arraystretch{1.3}
    \begin{array}{r|c|c|c|}
    \cline{2-4}
    \alpha(H) & \pm(\theta_j\pm\theta_k) & \pm\theta_j & \pm2\theta_j\\
    \cline{2-4}
    m_{\alpha}^+ & \beta & \beta(p-q) & \beta-1 \\
    \cline{2-4}
    m_{\alpha}^- & \beta & \beta(p-q) & \beta \\
    \cline{2-4}
    \end{array}
\end{equation}
Using \eqref{eq:compactk1ak2jac} with the above root system above we obtain the following Jacobian:
\begin{equation}\label{eq:newjacobi1}
    \prod_{j<k} |x_j - x_k|^{\beta} \prod_{j=1}^q x_j^{\frac{\beta(p-q+1)}{2}-1}(1-x_j)^{\frac{\beta-1}{2}} ,
\end{equation}
where $x_j = \sin^22\theta_j$ for all $j$. The $\beta=1$ case of \eqref{eq:newjacobi1} can be obtained from the CS decomposition approach too, with $(n+1)\times (n+1)$ orthogonal matrix and partitions $(p, q+1)$ and $(p+1, q)$. See Figure \ref{fig:manytomany}. The parameters of $\beta=2$ \eqref{eq:newjacobi1} cannot be obtained by the complex CSD, thus fall outside of the classical parameters. 

\subsection{Compact D$\RN{1}$-$\RN{3}$, A$\RN{2}$-$\RN{3}$}\label{sec:newjacobi2}
Another family of the $\kak$ decomposition arise from the following pairs of compact symmetric spaces. ($\beta=2, 4$)
\begin{align*}
    G/K_\tau, G/K_\sigma &= \ortho{2n}/\un{n}, \ortho{2n}/(\ortho{2p}\times \ortho{2q}) \\
    G/K_\tau, G/K_\sigma &= \un{2n}/\un{n, \mathbb{H}}, \un{2n}/(\un{2p}\times \un{2q}). 
\end{align*}
Under Cartan's classification they are types D$\RN{1}$-$\RN{3}$ and A$\RN{2}$-$\RN{3}$, respectively. The subgroup $A$ can be computed as
\begin{equation*}
    A = \{\text{$2n\times 2n$ matrices with the block structure} \begin{bsmallmatrix}
I_{p-q} & & & \\
 & C\otimes I_2 & & S\otimes J_1\\
& & I_{p-q} & \\
& S\otimes J_1 & & C \otimes I_2
\end{bsmallmatrix}\},
\end{equation*}
where $I_2$ is the $2\times 2$ identity matrix and $J_1 = \begin{bsmallmatrix}\,\,0 & 1 \\ -1 & 0\end{bsmallmatrix}$, and $C, S$ are $q\times q$ diagonal matrices with cosines and sines of $\theta_1, \dots, \theta_q$ on their diagonals. 
The root system is given as the following: ($\beta=2, 4$)
\begin{equation}\label{eq:newjacobiroots2}
    \renewcommand\arraystretch{1.2}
    \begin{array}{r|c|c|c|}
    \cline{2-4}
    \alpha(H) & \pm(\theta_j\pm\theta_k) & \pm\theta_j & \pm2\theta_j\\
    \cline{2-4}
    m_{\alpha}^+ & \beta & \frac{\beta}{2} (p-q) & \beta-1 \\
    \cline{2-4}
    m_{\alpha}^- & \beta & \frac{\beta}{2} (p-q) & \frac{\beta}{2}-1 \\
    \cline{2-4}
    \end{array}
\end{equation}
Again, using \ref{eq:compactk1ak2jac} with the root system above we obtain the following Jacobian, with the change of variables $x_j = \sin^2\theta_j$ for all $j$.
\begin{equation}\label{eq:newjacobi2}
    \prod_{j=1}^q x_j^{\frac{\beta(p-q+2)}{4}-1}(1-x_j)^{\frac{\beta-4}{4}} \prod_{j<k} |x_j - x_k|^{\beta}.
\end{equation}
They are $\beta=2,4$ Jacobi ensembles. Both cases could not be obtained from the classical CSD approach, so they are all non-classical parameters of the Jacobi ensemble. To see this at once, we compare three $\beta=2$ Jacobi densities each from Section \ref{sec:jacobibeta2}, \ref{sec:newjacobi1} and \ref{sec:newjacobi2}. Figure \ref{fig:newjacobi2} shows the possible parameters $\alpha_1, \alpha_2$ of the $\beta=2$ Jacobi ensemble obtained from each approach.

\section{Noncompact A$\RN{1}$, A, A$\RN{2}$: Hermite ensembles}\label{sec:hermite}

While Section \ref{sec:mixed} contains essentially new random matrix theories, Sections \ref{sec:hermite} and \ref{sec:Laguerre} review the Hermite and Laguerre ensembles for completeness \cite{altland1997nonstandard,caselle1996new,caselle2004random,ivanov2002random,zirnbauer1996riemannian}. 

\par The joint probability density of the Hermite ensemble is ($\beta=1,2,4$),
\begin{equation*}
    H_n^{(\beta)}(\lambda) \propto \prod_{j<k}|\lambda_j - \lambda_k|^\beta \prod_{j=1}^n e^{-\lambda^2_j/2}.
\end{equation*}
Hermite ensembles $\beta=1, 2, 4$ (GOE, GUE, GSE) arise as the eigenvalues of symmetric, Hermitian and self-dual Gaussian matrices. Hermite ensembles can be thought as the Gaussian measure endowed on the tangent space of noncompact symmetric spaces of the types A$\RN{1}$, A and A$\RN{2}$. The connection between these symmetric spaces and Hermite ensembles are made by Theorem \ref{thm:liealgebrajac}. The decomposition $\Psi$, \eqref{eq:liealgebramap} in Theorem \ref{thm:liealgebrajac} is the eigendecomposition of symmetric, Hermitian and self-dual matrices. The maximal abelian subalgebra $\mathfrak{a}$ is the collection of all real diagonal matrices, $\text{diag}(h_1, \dots, h_n)$. The restricted root system is the following ($1\le j<k\le n$). 
\begin{equation}\label{tab:hermiteroots}
\renewcommand\arraystretch{1.2}
\begin{array}{r|c|}
\cline{2-2}
\alpha(H) & \pm (h_j-h_k)\\
\cline{2-2}
m_{\alpha} & \beta\\
\cline{2-2}
\end{array}
\end{equation}

\subsection{Noncompact A$\RN{1}$, $\beta = 1$ GOE} 
The dual of the compact symmetric space type A$\RN{1}$, the noncompact symmetric space type A$\RN{1}$ is $G/K = \gl{n, \mathbb{R}}/\ortho{n}$, represented by the set $\mathcal{S}_n$ of all symmetric positive definite matrices. The tangent space at the identity of $\mathcal{S}_n$, $\mathfrak{p}$, is the set of all real symmetric matrices. The Gaussian measure on $\mathfrak{p}$ is, for $p\in\mathfrak{p}$, $\exp(-\tr (p^Tp)/2)dp$ where $dp$ is the Euclidean measure on $\mathfrak{p}$. From \eqref{eq:noncompactliealgebrajac} using \eqref{tab:hermiteroots} $\beta=1$ we obtain (integrate out $dk$)
\begin{equation*}
    \exp(-\tr (p^Tp)/2)dp \propto \prod_{j<k}|\lambda_j - \lambda_k|\prod_{j=1}^n e^{-\lambda_j^2/2}d\lambda_1\dots d\lambda_n,
\end{equation*}
for the eigenvalues of $p$, $\lambda_j = h_j$. This is the joint density of the GOE. 

\subsection{Noncompact A, $\beta = 2$ GUE} 
The noncompact symmetric space type A is $G/K = \gl{n, \mathbb{C}}/\un{n}$, represented by $\mathcal{H}_n$, the set of all Hermitian positive definite matrices. The tangent space at the identity of $\mathcal{H}_n$, $\mathfrak{p}$, is the set of all complex Hermitian matrices. The Gaussian measure on $\mathfrak{p}$ is, for $p\in\mathfrak{p}$, $\exp(-\tr (p^Hp)/2)dp$ where $dp$ is the (real) Euclidean measure on $\mathfrak{p}$. From \eqref{eq:noncompactliealgebrajac} using \eqref{tab:hermiteroots} $\beta=2$ we obtain
\begin{equation*}
    \exp(-\tr (p^Hp)/2)dp \propto \prod_{j<k}|\lambda_j - \lambda_k|^2\prod_{j=1}^n e^{-\lambda_j^2/2}d\lambda_1\dots d\lambda_n,
\end{equation*}
for the eigenvalues of $p$, $\lambda_j = h_j$. This is the joint density of the GUE. 

\subsection{Noncompact A$\RN{2}$, $\beta = 4$ GSE} 
The noncompact symmetric space type $A\RN{2}$ is $G/K = \gl{n, \mathbb{H}}/ \un{n, \mathbb{H}}$. We use $\un{n, \mathbb{H}}$ instead of $\symp{n}$ to clearly indicate the quaternionic realization. $G/K$ can be represented by the set of all quaternionic self-dual positive definite matrices, $\mathcal{QH}_n$. Again, the tangent space at the identity $\mathfrak{p}$ is the set of all quaternionic self-dual matrices. The Gaussian measure on $\mathfrak{p}$ is, for $p\in\mathfrak{p}$, $\exp(-\tr (p^Dp)/2)dp$ where $dp$ is the (real) Euclidean measure on $\mathfrak{p}$. From \eqref{eq:noncompactliealgebrajac} using \eqref{tab:hermiteroots} $\beta=4$ we obtain
\begin{equation*}
    \exp(-\tr (p^Dp)/2)dp \propto \prod_{j<k}|\lambda_j - \lambda_k|^4\prod_{j=1}^n e^{-\lambda_j^2/2}d\lambda_1\dots d\lambda_n,
\end{equation*}
for the eigenvalues of $p$, $\lambda_j = h_j$. This is the joint density of the GSE. 

\section{Noncompact BD$\RN{1}$, A$\RN{3}$, C$\RN{2}$: Laguerre ensembles}\label{sec:Laguerre}
The joint probability density of the Laguerre ensemble is ($\beta = 1,2,4$),
\begin{equation*}
    L_{\alpha, m}^{(\beta)}(\lambda) \propto \prod_{j<k}|\lambda_j - \lambda_k|^\beta \prod_{j=1}^m\lambda_j^{\alpha} e^{-\lambda_j/2}.
\end{equation*}
Laguerre ensembles $\beta = 1, 2, 4$ arise from Theorem \ref{thm:liealgebrajac} applied to noncompact symmetric spaces BD$\RN{1}$, A$\RN{3}$, C$\RN{2}$, D$\RN{3}$, BD, C, C$\RN{1}$. The last four cases of types D$\RN{3}$, BD, C, C$\RN{1}$ are well-studied in \cite{altland1997nonstandard} and we again omit these cases as discussed in Section \ref{sec:jacobi}. In particular, the first three symmetric spaces give the following Laguerre densities ($\beta = 1,2,4$ and $p\geq q$):
\begin{equation*}
    \prod_{j<k}|\lambda_j - \lambda_k|^\beta \prod_{j=1}^q \lambda_j^{\frac{\beta(p-q+1)}{2}-1} e^{-\lambda_j/2},
\end{equation*}
as these $\lambda_j$ values are the squared singular values of $p\times q$ i.i.d. Gaussian matrices. Equivalently, the eigenvalues of the matrix $A^\dagger A\in\mathbb{F}^{q\times q}$ are frequently used for sampling purpose, where $\dagger$ is the conjugate transposition. The tangent spaces of noncompact symmetric spaces of the types BD$\RN{1}$, A$\RN{3}$, C$\RN{2}$ are
\begin{equation}\label{eq:laguerretangentspace}
    \bigg\{\begin{bmatrix} 0 & X \\ X^\dagger & 0\end{bmatrix} : X \text{ is }p\times q\text{ matrix}\bigg\},
\end{equation}
and a choice of maximal abelian subalgebra $\mathfrak{a}$ is the set with $X$ being (nonsquare) diagonal matrix with diagonal elements $h_1, \dots, h_q$. The KAK decomposition $G = KAK$ of the noncompact symmetric spaces BD$\RN{1}$, A$\RN{3}$, C$\RN{2}$ is the \textit{hyperbolic CS decomposition} (HCSD) \cite{grimme1996model,higham2003jortho}. The decomposition $\mathfrak{p} = \cup_{k\in K} k\mathfrak{a}k^{-1}$ is the $p\times q$ SVD on upper right $p\times q$ corner. The restricted roots are the following ($\beta = 1, 2, 4$).
\begin{equation}\label{tab:laguerreroots}
    \renewcommand\arraystretch{1.2}
    \begin{array}{r|c|c|c|}
    \cline{2-4}
    \alpha(H) & \pm (h_j\pm h_k) & \pm h_j & \pm 2h_j\\
    \cline{2-4}
    m_{\alpha} & \beta  & \beta(p-q) & \beta-1\\
    \cline{2-4}
    \end{array}
\end{equation}

\subsection{Noncompact BD$\RN{1}$, $\beta=1$ Laguerre} 
The noncompact symmetric space type BD$\RN{1}$ is $G/K = \ortho{p, q}/(\ortho{p}\times\ortho{q})$. The tangent space $\mathfrak{p}$ \eqref{eq:laguerretangentspace} has the Gaussian measure as i.i.d. Gaussian distribution endowed on the elements of $X$. For $M\in\mathfrak{p}$ it is $\exp(-\tr(M^TM))d\mathfrak{p}$. From \eqref{eq:noncompactliealgebrajac} using \eqref{tab:laguerreroots} $\beta=1$ we obtain
\begin{equation*}
    \exp(-\tr(M^TM))d\mathfrak{p} \propto \prod_{j<k}|\lambda_j-\lambda_k|\prod_{j=1}^q e^{-\lambda_j/2} \lambda_j^{\frac{p-q-1}{2}}d\lambda_1\cdots d\lambda_q,
\end{equation*}
with the change of variables $\lambda_j = h_j^2$. Thus the values $\lambda_1, \dots, \lambda_q$ are the squared singular values of the upper right corner of $M$. The obtained measure is the joint density of the $\beta=1$ Laguerre ensemble. 

\subsection{Noncompact A$\RN{3}$, $\beta = 2$ Laguerre} 
The noncompact symmetric space type A$\RN{3}$ is $G/K = \un{p, q}/(\un{p}\times\un{q})$. The tangent space \eqref{eq:laguerretangentspace} has the Gaussian measure as i.i.d. complex Gaussian distribution endowed on the elements of $X$. For $M\in\mathfrak{p}$ that is $\exp(-\tr(M^HM))d\mathfrak{p}$. From \eqref{eq:noncompactliealgebrajac} using \eqref{tab:laguerreroots} $\beta=2$ we obtain
\begin{equation*}
    \exp(-\tr(M^HM))d\mathfrak{p} \propto \prod_{j<k}|\lambda_j-\lambda_k|^2\prod_{j=1}^q e^{-\lambda_j/2}\lambda_j^{p-q}d\lambda_1\cdots d\lambda_q,
\end{equation*}
with the change of variables $\lambda_j = h_j^2$. Again the values $\lambda_1, \dots, \lambda_q$ are the squared singular values of the upper right corner of $M$. The obtained measure is the joint density of the $\beta=2$ Laguerre ensemble.  

\subsection{Noncompact C$\RN{2}$, $\beta = 4$ Laguerre} 
The noncompact symmetric space C$\RN{2}$ is $G/K = \un{p, q,\mathbb{H}}/(\un{p,\mathbb{H}}\times\un{q,\mathbb{H}})$. The tangent space \eqref{eq:laguerretangentspace} has the Gaussian measure as i.i.d. quaternionic Gaussian distribution endowed on the elements of $X$. For $M\in\mathfrak{p}$ that is $\exp(-\tr(M^DM))d\mathfrak{p}$. From \eqref{eq:noncompactliealgebrajac} using \eqref{tab:laguerreroots} $\beta=4$ we obtain
\begin{equation*}
    \exp(-\tr(M^DM))d\mathfrak{p} \propto \prod_{j<k}|\lambda_j-\lambda_k|^4\prod_{j=1}^q e^{-\lambda_j/2}\lambda_j^{2(p-q)+1}d\lambda_1\cdots d\lambda_q,
\end{equation*}
with the change of variables $\lambda_j = h_j^2$. The values $\lambda_1, \dots, \lambda_q$ are the squared singular values of the upper right corner of $M$. The obtained measure is the joint density of the $\beta=4$ Laguerre ensemble.

\section*{Acknowledgements}
We thank Martin Zirnbauer for the lengthy email thread from 2001, where he patiently explained which random matrix ensembles seemed to be covered by symmetric spaces. We thank Eduardo Due{\~n}ez for another lengthy email thread back in 2013. We thank Pavel Etingof for suggesting the $\kak$ decomposition and pointing us to key references, Bernie Wang  for so very much and the Fall 2020 Random Matrix Theory class (MIT 18.338) for valuable suggestions. We also thank Sigurður Helgason for lively discussions by email. We thank NSF grants OAC-1835443, OAC-2103804, SII-2029670, ECCS-2029670, PHY-2021825 for financial support.

\bibliographystyle{plain}
\bibliography{bibliography.bib}

\end{document}